\documentclass[11pt]{scrartcl}
\usepackage[utf8]{inputenc}
\usepackage{hyperref}
\usepackage{amsmath,amssymb,amstext}
\usepackage{braket}
\usepackage{amsthm}
\usepackage{dsfont}
\usepackage{color}
\usepackage[title]{appendix}

\newtheorem{vor}{Assumption}[section]
\newtheorem{theorem}[vor]{Theorem}

\newtheorem{lem}[vor]{Lemma}
\newtheorem{cor}[vor]{Corollary}

\theoremstyle{definition}
\newtheorem{defi}[vor]{Definition}

\numberwithin{equation}{section}

\begin{document}
\title{The Fr\"ohlich Polaron at Strong Coupling -- Part II: Energy-Momentum Relation and Effective Mass}
\author{\vspace{.3cm} \textsc{Morris Brooks and Robert Seiringer}\\ IST Austria, Am Campus 1, 3400 Klosterneuburg, Austria}
\date{October 20, 2022}
\maketitle

\begin{abstract} 
\textsc{Abstract}. We study the Fr\"ohlich polaron model in $\mathbb{R}^3$, and prove a lower bound on its ground state energy as a function of the total momentum. The bound is asymptotically sharp at large coupling. In combination with a corresponding upper bound proved earlier \cite{MMS}, it shows that the energy is approximately parabolic below the continuum threshold, and that the polaron's effective mass (defined as the semi-latus rectum of the parabola) is given by the celebrated Landau--Pekar formula. In particular, it diverges as $\alpha^4$ for large coupling constant $\alpha$.
\end{abstract}

\section{Introduction and Main Results}
\label{Section: Introduction}
This is the second part of a study of the Fr\"ohlich polaron \cite{F37} in the regime of strong coupling between the electron and the phonons, which are the optical modes of a polar crystal. Our goal is to quantify the heuristic picture that the mass of an electron in a polarizable medium effectively increases due to an emerging phonon cloud attached to it. We are going to verify that the energy-momentum relation of a polaron is asymptotically given by the semi-classical formula $E(P)-E(0)=\frac{|P|^2}{2\alpha^4 m}$, which agrees with the energy-momentum relation of a particle having mass $\alpha^4 m$, where $\alpha^4 m$ is the  asymptotic formula conjectured by Landau and Pekar  \cite{Lpekar} for the mass of a polaron in the regime where the coupling parameter $\alpha$ goes to infinity. 

Following the notation of the first part \cite{BS1}, where a second order expansion for the absolute ground state energy of a polaron was verified, we are going to use creation and annihilation operators satisfying the semi-classical rescaled canonical commutation relations $[a(f),a^\dagger(g)]=\alpha^{-2}\braket{f|g}$ for $f,g\in L^2\! \left(\mathbb{R}^3\right)$, in order to introduce the Fr\"ohlich Hamiltonian acting on the Fock space $L^2\!\left(\mathbb{R}^3\right)\otimes \mathcal{F}\left(L^2\!\left(\mathbb{R}^3\right)\right)$ as
\begin{align*}
\mathbb{H}:=-\Delta_x-a\left(w_x\right)-a^\dagger\left(w_x\right)+\mathcal{N},
\end{align*}
where $w_x(x'):=\pi^{-\frac{3}{2}}|x'-x|^{-2}$ and the (rescaled) particle number operator $\mathcal{N}$ equals $\mathcal{N}:=\sum_{n=1}^\infty a^\dagger(\varphi_n)a(\varphi_n)$ for an orthonormal basis $\{\varphi_n:n\in \mathbb{N}\}$ of $L^2(\mathbb{R}^3)$. The Fr\"ohlich Hamiltonian $\mathbb{H}$ commutes with the components $\left(\mathbb{P}_1,\mathbb{P}_2,\mathbb{P}_3\right)$ of the total momentum operator
%
$$
\mathbb{P}:=\frac{1}{i}\nabla+\alpha^2\int_{\mathbb{R}^3}k\,  a_k^\dagger a_k\mathrm{d}k \,,
$$
 where we use the standard notation $\int_{\mathbb{R}^3}f(k) a_k^\dagger a_k\mathrm{d}k$ as a symbolic expression for the operator $\sum_{n,m=1}^\infty \Big\langle \varphi_n\Big|f\left(\frac{1}{i}\nabla \right)\Big|\varphi_m\Big\rangle a^\dagger(\varphi_n)a(\varphi_m)$. 
 Hence we can study  their joint spectrum $\sigma\left(\mathbb{P},\mathbb{H}\right)\subseteq \mathbb{R}^4$, and  define the ground state energy $E_\alpha(P)$ of $\mathbb{H}$ at total momentum $P$ as 
 $E_\alpha(P):=\inf\big\{E:(P,E)\in \sigma\left(\mathbb{P},\mathbb{H}\right)\big\}$. Our main result below is the proof of  the asymptotic energy-momentum relation
\begin{align}
\label{Equation-EnergyMomentum}
E_\alpha(P)=E_\alpha(0)+\min\bigg\{\frac{|P|^2}{2\alpha^4 m}, \alpha^{-2}\bigg\}+O_{\alpha\rightarrow \infty}\left(\alpha^{-(2+w)}\right),
\end{align}
where $w>0$ is a suitable constant and $m$ is the conjectured constant by Landau and Pekar. In order to provide an explicit expression for $m$, let us first define the Pekar functional $\mathcal{F}^\mathrm{Pek}\! \left(\varphi\right):=\|\varphi\|^2+\inf \sigma\left(-\Delta+V_\varphi\right)$ for $\varphi\in L^2\! \left(\mathbb{R}^3\right)$, where we define the potential $V_\varphi:=-2\left(-\Delta\right)^{-\frac{1}{2}}\mathfrak{Re}\, \varphi$. If follows from the analysis in \cite{Li} that there exists a unique radial minimizer $\varphi^\mathrm{Pek}$ of the functional $\mathcal{F}^\mathrm{Pek}$. With this minimizer at hand, we can  introduce the constant $m:=\frac{2}{3}\left\|\nabla \varphi^\mathrm{Pek}\right\|^2$ in Eq.~(\ref{Equation-EnergyMomentum}). 

In order to formulate our main Theorem \ref{Theorem: Parabolic Lower Bound}, let us further introduce the minimal Pekar energy $e^\mathrm{Pek}:=\inf_\varphi \mathcal{F}^\mathrm{Pek}\! \left(\varphi\right)$ as well as the Hessian $H^\mathrm{Pek}$ of $\mathcal{F}^\mathrm{Pek}$ at the minimizer $\varphi^\mathrm{Pek}$ restricted to real-valued functions $\varphi\in L^2_\mathbb{R}\! \left(\mathbb{R}^3\right)$, i.e. we define $H^\mathrm{Pek}$ as the unique self-adjoint operator on $L^2\! \left(\mathbb{R}^3\right)$  satisfying 
\begin{align*}
\braket{\varphi|H^\mathrm{Pek}|\varphi}=\lim_{\epsilon\rightarrow 0}\frac{1}{\epsilon^2}\left(\mathcal{F}^\mathrm{Pek}\! \left(\varphi^\mathrm{Pek}+\epsilon \varphi\right)-e^\mathrm{Pek}\right)
\end{align*}
for all $\varphi\in L^2_\mathbb{R}\! \left(\mathbb{R}^3\right)$. With this notation at hand, we can state our main new result in Theorem \ref{Theorem: Parabolic Lower Bound}. It provides a sharp asymptotic lower bound on the ground state energy $E_\alpha(P)$ of the operator $\mathbb{H}$ as a function of the total momentum $\mathbb{P}$.

\begin{theorem}
\label{Theorem: Parabolic Lower Bound}
There exists a constant $w>0$ such that 
\begin{align}
\label{Equation-Main}
E_\alpha(P)\geq e^\mathrm{Pek}-\frac{1}{2\alpha^2}\mathrm{Tr}\left[1-\sqrt{H^\mathrm{Pek}}\, \right]+\min\bigg\{\frac{|P|^2}{2\alpha^4 m}, \alpha^{-2}\bigg\}-\alpha^{-(2+w)}
\end{align}
for all $P\in \mathbb{R}^3$ and for all $\alpha\geq \alpha_0$, where $\alpha_0$ is a suitable constant.
\end{theorem}

That the lower bound in Eq.~(\ref{Equation-Main}) is indeed sharp follows from the corresponding asymptotic upper bound established in \cite{MMS}, given by 
\begin{align}
\label{Equation-Parabolic Upper Bound}
E_\alpha(P)\leq e^\mathrm{Pek}-\frac{1}{2\alpha^2}\mathrm{Tr}\left[1-\sqrt{H^\mathrm{Pek}}\, \right]+\min\bigg\{\frac{|P|^2}{2\alpha^4 m}, \alpha^{-2}\bigg\}+C_{\epsilon}\alpha^{-\frac{5}{2}+\epsilon},
\end{align} 
where $\epsilon>0$ is arbitrary and $C_{\epsilon}$ a suitable constant. In combination with Eq.~(\ref{Equation-Main}) this shows that 
\begin{align*}
E_\alpha(P)=e^\mathrm{Pek}-\frac{1}{2\alpha^2}\mathrm{Tr}\left[1-\sqrt{H^\mathrm{Pek}}\, \right]+\min\bigg\{\frac{|P|^2}{2\alpha^4 m}, \alpha^{-2}\bigg\}+O_{\alpha\rightarrow \infty}\left(\alpha^{-(2+w)}\right)
\end{align*}
for all $P\in \mathbb{R}^3$, which in particular proves Eq.~(\ref{Equation-EnergyMomentum}). Note that  $\alpha^{-2}$ corresponds to the continuum threshold; i.e., $\sigma(\mathbb{P},\mathbb{H}) \supset \mathbb{R}^3\times [E_\alpha(0) + \alpha^{-2},\infty)$, the latter corresponding to states describing free phonons on top of the polaron ground state \cite{JSM,LMM}. 

In particular, $E_\alpha(P)$ has an approximate parabolic shape below the continuum threshold, i.e., for $|P| < \sqrt{2m} \alpha$. The Landau--Pekar formula for the effective mass appears in the limit $\alpha\to \infty$ as the semi-latus rectum of the parabola, in the sense that  for any $0<|P|< \sqrt{2m}$
\begin{align}
\label{Equation: Alternative Effective Mass}
m= 
\lim_{\alpha\rightarrow \infty}\alpha^{-4}\frac{|\alpha P|^2}{2\left(E_\alpha(\alpha P)-E_\alpha(0)\right)} \,.
\end{align}
It is common the define the polaron's effective mass for fixed $\alpha$ as $$
M_\mathrm{eff}(\alpha):=\lim_{P\rightarrow 0}\frac{|P|^2}{2\left(E_\alpha(P)-E_\alpha(0)\right)}\,.
$$
The quantity on the right hand side of Eq.~(\ref{Equation: Alternative Effective Mass}) is clearly related to the large $\alpha$ limit of $\alpha^{-4}M_\mathrm{eff}(\alpha)$, with the difference being that the limit $P\rightarrow 0$ is taken before the limit $\alpha\rightarrow \infty$. While it is not clear at this point how to obtain the lower bound $\lim_{\alpha\rightarrow \infty}\alpha^{-4}M_\mathrm{eff}(\alpha)\geq m$, we can make use of the inequality $E_\alpha(P)\leq E_\alpha(0)+\frac{|P|^2}{2 M_\mathrm{eff}(\alpha)}$ recently proved in \cite{Po} in order to verify the upper bound $\lim_{\alpha\rightarrow \infty}\alpha^{-4}M_\mathrm{eff}(\alpha)\leq m$. In fact, by applying Eq.~(\ref{Equation-EnergyMomentum}) in the special case of $P$ satisfying $|P| = \sqrt{2m}\alpha$ we have
\begin{align*}
E_\alpha(0)+\frac{1}{\alpha^2 }+O_{\alpha\rightarrow \infty}\left(\alpha^{-(2+w)}\right)=E_\alpha(P)\leq E_\alpha(0)+\frac{m\alpha^2}{M_\mathrm{eff}(\alpha)},
\end{align*}
which yields the claimed upper bound on $M_\mathrm{eff}(\alpha)$. We formulate it as the subsequent Corollary. 

\begin{cor}
\label{Corollary: Parabolic Lower Bound}
There exists a constant $w>0$ such that $M_\mathrm{eff}(\alpha)\leq \alpha^4 m+O_{\alpha\rightarrow \infty}\! \left(\alpha^{4-w}\right)$.
\end{cor}

The remainder of this paper contains the proof of Theorem~\ref{Theorem: Parabolic Lower Bound}. In order to guide the reader, we start with a short explanation of the main strategy. 

\medskip

\textbf{Proof strategy of Theorem \ref{Theorem: Parabolic Lower Bound}.} Since $(P,E_\alpha(P))$ is an element of the joint spectrum of the operator pair $(\mathbb{P},\mathbb{H})$, there clearly exist states $\Psi_\alpha$ satisfying $\mathbb{P}\Psi_\alpha\approx P\Psi_\alpha$ and $\mathbb{H}\Psi_\alpha\approx E_\alpha(P)\Psi_\alpha$. In order to verify Theorem \ref{Theorem: Parabolic Lower Bound}, it is therefore  enough to show that $\braket{\Psi_\alpha|\mathbb{H}|\Psi_\alpha}$ is bounded from below by the right hand side of Eq.~(\ref{Equation-Main}). For this to hold it is crucial to use the additional information  $\mathbb{P}\Psi_\alpha\approx P\Psi_\alpha$ on the momentum, since in general $\mathbb{H}$, as an operator, is not bounded from below by the right hand side of Eq.~(\ref{Equation-Main}). 
It is not possible to transform the constrained minimization problem to a global one by the usual method of Lagrange multipliers, since the operators $\mathbb{P}$ are not bounded relative to $\mathbb{H}$. More precisely, while clearly
\begin{align}
\label{Equation: Trivial Lower Bound}
E_\alpha(P)\geq \inf \sigma\! \left(\mathbb{H}+\lambda(P-\mathbb{P})\right)
\end{align}
for any $\lambda  \in \mathbb{R}^3$, such a bound is insufficient as the right hand side is $-\infty$ for $\lambda\neq 0$, which follows easily from the fact that 
 $E_\alpha(P)$ is bounded uniformly in $P$ (compare with Eq.~(\ref{Equation-EnergyMomentum})).

In order to improve the lower bound in Eq.~(\ref{Equation: Trivial Lower Bound}), we introduce a wavenumber cut-off $\Lambda$ in the Hamiltonian $\mathbb{H}$ as well as in the momentum operator $\mathbb{P}$, leading to the study of the ground state energy $E_{\alpha,\Lambda}(P)$ of the truncated Hamiltonian $\mathbb{H}_\Lambda$ as a function of the truncated momentum $\mathbb{P}_\Lambda$. As we will show in the subsequent Section \ref{Section: Reduction to bounded Wavenumbers}, it is enough to prove Eq.~(\ref{Equation-Main}) for the modified energy $E_{\alpha,\Lambda}(P)$ in order to verify our main Theorem \ref{Theorem: Parabolic Lower Bound}. By introducing the cut-off we  manually exclude the radiative regime where a single phonon carries the total momentum, which is responsible for the (approximately) flat energy-momentum relation  $E_\alpha(P)$ above the threshold $|P|=\sqrt{2m}\alpha$ and the resulting collapse of the quadratic approximation $E_\alpha(P)-E_\alpha(0)\approx  \frac{|P|^2}{2\alpha^4 m}$ above this threshold.

In contrast, in the presence of the cut-off, it turns out that we can apply the method of Lagrange multiplies. 
%
%
%
%
We shall follow the strategy developed in the first part  \cite{BS1}, and construct approximate eigenstates $\Psi_\alpha$ to the joint eigenvalue $(P,E_{\alpha,\Lambda}(P))$ of the operator pair $(\mathbb{P}_\Lambda, \mathbb{H}_\Lambda)$, which in addition satisfy (complete) Bose--Einstein condensation with respect to the minimizer $\varphi^\mathrm{Pek}$ of the Pekar functional $\mathcal{F}^\mathrm{Pek}$. In this context we call $\Psi_\alpha$ an approximate eigenstate in case $\braket{\Psi_\alpha|(\mathbb{P}_\Lambda-P)^2|\Psi_\alpha}=O_{\alpha\rightarrow \infty}\! \left(\alpha^{2-r}\right)$ and $E_{\alpha,\Lambda}(P)\geq \braket{\Psi_\alpha|\mathbb{H}_\Lambda|\Psi_\alpha}+O_{\alpha\rightarrow \infty}\! \left(\alpha^{-(2+r)}\right)$ for some $r>0$. 
In order to verify that $E_{\alpha,\Lambda}(P)$ is bounded from below by the right hand side of Eq.~(\ref{Equation-Main}), it is consequently enough to show that
\begin{align}
\label{Equation: Central Estimate}
\Big\langle\Psi\Big|\mathbb{H}_\Lambda+\lambda\Big(P-\mathbb{P}_\Lambda\Big)\Big|\Psi\Big\rangle\geq e^\mathrm{Pek}-\frac{1}{2\alpha^2}\mathrm{Tr}\left[1-\sqrt{H^\mathrm{Pek}}\, \right]+\lambda P-\frac{\alpha^4 m |\lambda|^2}{2}-\alpha^{-(2+w)}
\end{align}
for all states $\Psi$ satisfying (complete) Bose--Einstein condensation with respect to the minimizer $\varphi^\mathrm{Pek}$, providing the desired lower bound for the optimal choice $\lambda=\frac{P}{m\alpha^4}$, with the term $\frac{\alpha^4 m |\lambda|^2}{2}$ in Eq.~(\ref{Equation: Central Estimate}) arising naturally as the Legendre transformation of the quadratic approximation $\frac{|P|^2}{2\alpha^4 m}$. 

Since Eq.~(\ref{Equation: Central Estimate}) claims a global lower bound, i.e. there is no constraint on the momentum of $\Psi$, we can utilize the methods developed in the first part \cite{BS1}, where  a lower bound on the total minimum $E_\alpha=\inf \sigma(\mathbb{H})$ was established. The basic idea is that we can find, up to a unitary transformation, a lower bound on the operator $\mathbb{H}_\Lambda+\frac{P}{m\alpha^4}\Big(P-\mathbb{P}_\Lambda\Big)$ of the form $e^\mathrm{Pek}+\frac{|P|^2}{2\alpha^4 m}+\mathbb{Q}+O_{\alpha\rightarrow \infty}\! \left(\alpha^{-(2+r)}\right)$, where $\mathbb{Q}$ is a system of harmonic oscillators, which holds when tested against states satisfying (complete) Bose--Einstein condensation. The ground state energy of $\mathbb{Q}$ can then be computed explicitly, giving rise to the quantum correction $-\frac{1}{2\alpha^2}\mathrm{Tr}\left[1-\sqrt{H^\mathrm{Pek}}\, \right]$ in Eq.~(\ref{Equation-Main}).\\

\textbf{Outline.} The paper is structured as follows. In Section \ref{Section: Reduction to bounded Wavenumbers} we shall show that it is sufficient to prove Eq.~(\ref{Equation-Main}) for a model including a suitable ultraviolet wavenumber cut-off in order to verify our main Theorem \ref{Theorem: Parabolic Lower Bound}. In the subsequent Section \ref{Section: Construction of a Condensate}, we will  construct approximate eigenstates for the truncated model defined in Section \ref{Section: Reduction to bounded Wavenumbers}, which in addition satisfy (complete) Bose--Einstein condensation with respect to the state $\varphi^\mathrm{Pek}$. Section \ref{Section: Proof of Theorem} is then devoted to the proof of our main technical Theorem \ref{Theorem: Parabolic Lower Bound for the truncated Model}, where we use the method of Lagrange multipliers in order to get rid of the momentum constraint. Finally, Appendix \ref{Appedinx: Auxiliary Results} contains auxiliary results on commutator estimates as well as properties of the Pekar minimizer $\varphi^\mathrm{Pek}$, which get used in the proof.

\section{Reduction to Bounded Wavenumbers}
\label{Section: Reduction to bounded Wavenumbers}
In this section we shall introduce the truncated Hamiltonian $\mathbb{H}_\Lambda$, which includes a wavenumber restriction $|k|\leq \Lambda$, and we are going to state our main technical Theorem \ref{Theorem: Parabolic Lower Bound for the truncated Model}, which provides an analogue of Theorem \ref{Theorem: Parabolic Lower Bound} for the truncated model. While the proof of Theorem \ref{Theorem: Parabolic Lower Bound for the truncated Model} is the content of Sections~\ref{Section: Construction of a Condensate} and~\ref{Section: Proof of Theorem}, we will verify in this Section that Theorem \ref{Theorem: Parabolic Lower Bound} is a consequence of Theorem \ref{Theorem: Parabolic Lower Bound for the truncated Model}, i.e. we will explain why it is enough to prove Eq.~(\ref{Equation-Main}) for a model including a wavenumber regularization. The quantum nature of our system, and in particular the discrete spectrum $\sigma\left(\mathcal{N}\right)=\Big\{0,\frac{1}{\alpha^2},\frac{2}{\alpha^2},\dots \Big\}$ of the number operator $\mathcal{N}$, is essential for this argument to work. In contrast, 
in the classical case the effective mass is infinite since there nothing prevents a priori the wavenumber from escaping to infinity without an energy penalty, and 
one has to introduce a suitable regularization in order to observe the expected asymptotics $M_\mathrm{eff}=\alpha^4 m+o_{\alpha\rightarrow \infty}\left(\alpha^4\right)$, see \cite{FRS}.

Before formulating Theorem \ref{Theorem: Parabolic Lower Bound for the truncated Model}, we shall introduce some useful notation. Following \cite{BS1}, we define for a function $f:X\longrightarrow \mathbb{R}$, $\epsilon\geq 0$ and $-\infty\leq a\leq b\leq \infty$, the function $\chi^\epsilon\left(a\leq f\leq b\right):X\longrightarrow [0,1]$ as
\begin{align}
\label{Equation-Epsilon cut-off}
\chi^\epsilon\left(a\leq f(x)\leq b\right):=
\begin{cases}
\alpha\left(\frac{f(x)-b}{\epsilon}\right)\beta\left(\frac{f(x)-a}{\epsilon}\right),\text{ for }\epsilon>0\\
\mathds{1}_{[a,b]}\left(f(x)\right),\text{ for }\epsilon=0, 
\end{cases}
\end{align}
where $\alpha,\beta:\mathbb{R}\longrightarrow [0,1]$ are given $C^\infty$ functions such that $\alpha^2+\beta^2=1$, $\mathrm{supp}\left(\alpha\right)\subset(-\infty,1)$ and $\mathrm{supp}\left(\beta\right)\subset(-1,\infty)$. Similarly we define the operator $\chi^\epsilon\left(a\leq T\leq b\right):=\int \chi^\epsilon\left(a\leq t\leq b\right)\mathrm{d}E$, where $T$ is a self-adjoint operator and $E$ the corresponding spectral measure. Furthermore let us write $\chi(a\leq f\leq b)$ in case $\epsilon=0$ and $\chi^\epsilon \left(\cdot\leq b\right)$, respectively $\chi^\epsilon \left(a\leq \cdot\right)$, in case $a=-\infty$ or $b=\infty$, respectively. With this notation at hand, we define the Hamiltonian $\mathbb{H}_\Lambda$ with wavenumber cut-off $\Lambda\geq 0$ as
\begin{align}
\label{Equation-CutOffHamilton}
\mathbb{H}_\Lambda:=-\Delta_x-a\left(\chi\left(\left|\nabla\right|\leq \Lambda\right)w_x\right)-a^\dagger\left(\chi\left(\left|\nabla\right|\leq \Lambda\right)w_x\right)+\mathcal{N}.
\end{align}

\begin{theorem}
\label{Theorem: Parabolic Lower Bound for the truncated Model}
Let $E_{\alpha,\Lambda}(P)$ be the ground state energy of the operator $\mathbb{H}_\Lambda$ as a function of the (one-component of the) truncated total momentum 
$$
\mathbb{P}_\Lambda:=\frac{1}{i}\nabla_{x_1}+\alpha^2\int \! \chi^1\! \left(\Lambda^{-1}|k_1|\leq 2\right)\! k_1\,  a_k^\dagger a_k\mathrm{d}k
$$ 
and let $\Lambda=\alpha^{\frac{4}{5}(1+\sigma)}$ with $0<\sigma<\frac{1}{9}$. Then there exists a constant $w>0$ such that for all $C>0$, $|P|\leq C\alpha$ and $\alpha_0\geq \alpha(\sigma, C)$ 
\begin{align}
\label{Equation-Main with cut-off}
E_{\alpha,\Lambda}(P)\geq e^\mathrm{Pek}-\frac{1}{2\alpha^2}\mathrm{Tr}\left[1-\sqrt{H^\mathrm{Pek}}\, \right]+\frac{|P|^2}{2\alpha^4 m}-\alpha^{-(2+w)},
\end{align}
where $\alpha_0(\sigma,C)$ is a suitable constant.
\end{theorem}

For technical reasons  we use here the smooth cut-off $\chi^1\! \left(\Lambda^{-1}|k_1|\leq 2\right)$ instead of the sharp cut-off $\chi \left(\Lambda^{-1}|k_1|\leq 1\right)$ in the definition of the momentum operator $\mathbb{P}_\Lambda$. Note also that the momentum cut-off appears in \eqref{Equation-CutOffHamilton} only in the interaction term, and not in the field energy $\mathcal{N}$. 
In the following we shall  argue that, as a consequence of Theorem \ref{Theorem: Parabolic Lower Bound for the truncated Model}, Eq.~(\ref{Equation-Main with cut-off}) is also valid with $\mathbb{P}_\Lambda$ replaced by  $\mathbb{P}'_1:=\frac{1}{i}\nabla_{x_1}+\alpha^2\int_{|k|\leq \Lambda} k_j\,  a_k^\dagger a_k\mathrm{d}k$ having the sharp cut-off, and with $H_\Lambda$ replaced by the fully restricted Hamiltonian $\mathbb{H}'_\Lambda:=\mathbb{H}_\Lambda-\int_{|k|>\Lambda}a_k^\dagger a_k \mathrm{d}k$. In order to see this, observe that $\mathbb{P}'_1$ and $\mathbb{H}'_\Lambda$ are the restrictions (in the sense of operators) of $\mathbb{P}_\Lambda$ and $\mathbb{H}_\Lambda$ to states of the form $\Psi'\otimes \Omega$, where $\Psi'\in L^2\! \bigg(\mathbb{R}^3,\mathcal{F}\Big(\mathrm{ran} \chi\big(|\nabla|\leq \Lambda\big)\Big)\bigg)$ and $\Omega$ is the vacuum in $\mathcal{F}\Big(\mathrm{ran} \chi\big(|\nabla|> \Lambda\big)\Big)$. Hence
\begin{align*}
\sigma\left(\mathbb{P}'_1,\mathbb{H}_\Lambda'\right)\subseteq \sigma\left(\mathbb{P}_\Lambda,\mathbb{H}_\Lambda\right),
\end{align*}
and therefore we obtain as an immediate consequence of the previous Theorem \ref{Theorem: Parabolic Lower Bound for the truncated Model} that
\begin{align}
\label{Equation-EnergyMomentumRelation}
E\geq e^\mathrm{Pek}-\frac{1}{2\alpha^2}\mathrm{Tr}\left[1-\sqrt{H^\mathrm{Pek}}\, \right]+\frac{|P|^2}{2\alpha^4 m}-\alpha^{-(2+w)}
\end{align}
for all $(P,E)\in \sigma\left(\mathbb{P}'_1,\mathbb{H}_\Lambda'\right)$ with $|P|\leq C\alpha$ and $\alpha\geq \alpha_0(\sigma,C)$. In the proof of Theorem~\ref{Theorem: Parabolic Lower Bound} below it will be useful to have Eq.~(\ref{Equation-EnergyMomentumRelation}) for $\mathbb{P}'_1$ and $\mathbb{H}'_\Lambda$, instead of Eq.~(\ref{Equation-Main with cut-off}) for $\mathbb{P}_\Lambda$ and $\mathbb{H}_\Lambda$.

In order to verify Theorem \ref{Theorem: Parabolic Lower Bound}, it is convenient to introduce the ground state energy $E^*_{\alpha,\Lambda}(P)$ of the operator $\mathbb{H}_\Lambda$ as a function of $\mathbb{P}$. Note that in contrast to $E_{\alpha,\Lambda}(P)$, we do not use a wavenumber cut-off in the momentum operator here, while we still have the cut-off in the Hamiltonian $\mathbb{H}_\Lambda$. In the following Lemma \ref{Lemma-GroundStateEnergyComparison} we are going to utilize the results in \cite{FS,M}, where  the energy cost of introducing a wavenumber cut-off in the Hamiltonian is quantified, in order to compare $E^*_{\alpha,\Lambda}(P)$ with $E_{\alpha}(P)$.

\begin{lem}
\label{Lemma-GroundStateEnergyComparison}
Let $\Lambda=\alpha^{\frac{4}{5}(1+\sigma)}$ for $\sigma>0$. Then there exists a constant $C'>0$, such that for all $P\in \mathbb{R}^3$ and $\alpha$ large enough
\begin{align*}
E_\alpha(P)\geq E^*_{\alpha,\Lambda}(P)-C' \alpha^{-2(1+\sigma)}.
\end{align*}
\end{lem}
\begin{proof}
By the results in \cite{FS,M}, there exists a $C>0$ such that for $\alpha$ large enough
\begin{align}
\label{Equation-UltravioletCutOff}
\mathbb{H}_\Lambda\leq \mathbb{H} +C\alpha^{-2(1+\sigma)}\left(\mathbb{H}^2+1\right).
\end{align}
This was first shown in \cite{FS} for a confined polaron model on a bounded domain, but the method extends in a straightforward way to the model on $\mathbb{R}^3$, as shown in \cite{M} (see also \cite{FeS} for the corresponding result for a polaron model on a torus). 
In the following, let $\Psi_\epsilon$ be a state satisfying $\chi\left(\sum_{j=1}^3\left(\mathbb{P}_j-P_j\right)^2\leq \epsilon^2\right)\Psi_\epsilon=\Psi_\epsilon$ and $\braket{\Psi_\epsilon|\big(\mathbb{H}-E_\alpha(P)\big)^2|\Psi_\epsilon}\leq \epsilon^2$, where $\epsilon>0$. By Eq.~(\ref{Equation-UltravioletCutOff}) we therefore have
\begin{align*}
\braket{\Psi_\epsilon|\mathbb{H}_\Lambda|\Psi_\epsilon}&\leq E_\alpha(P)+C\alpha^{-2(1+\sigma)}\left(\braket{\Psi_\epsilon|\mathbb{H}^2|\Psi_\epsilon}+1\right)+\epsilon\\
&\leq E_\alpha(P)+C\alpha^{-2(1+\sigma)}\left(2E_\alpha(P)^2+2\epsilon^2+1\right)+\epsilon\leq E_\alpha(P)+C'\alpha^{-2(1+\sigma)}+\epsilon
\end{align*}
for $0<\epsilon\leq 1$ and a suitable $C'$, where we used that $E_\alpha(P)$ is uniformly bounded for $P\in \mathbb{R}^3$ and $\alpha\geq 1$ in the last inequality. Hence $\chi\left(\mathbb{H}_\Lambda\leq E_\alpha(P)+C'\alpha^{-2(1+\sigma)}+\epsilon\right)\Psi_\epsilon \neq 0$. Using $\chi\left(\sum_{j=1}^3\left(\mathbb{P}_j-P_j\right)^2\leq \epsilon^2\right)\Psi_\epsilon=\Psi_\epsilon$, we obtain
\begin{align*}
A_\epsilon:=\sigma\left(\mathbb{P},\mathbb{H}_\Lambda\right)\cap \left(B_{\epsilon}(P)\times \left(-\infty,E_\alpha(P)+C'\alpha^{-2(1+\sigma)}+\epsilon\right]\right)\neq \emptyset.
\end{align*}
Since $\mathbb{H}_\Lambda$ is bounded from below, $(A_\epsilon)_{0<\epsilon\leq 1}$ is a monotone sequence of non-empty compact sets, i.e. $A_{\epsilon_1}\subseteq A_{\epsilon_2}$ for $\epsilon_1\leq \epsilon_2$, and consequently
\begin{align*}
\sigma\left(\mathbb{P},\mathbb{H}_\Lambda\right)\cap \left(\{P\}\times \left(-\infty,E_\alpha(P)+C'\alpha^{-2(1+\sigma)}\right]\right)=\bigcap_{0<\epsilon\leq 1}A_\epsilon\neq \emptyset,
\end{align*}
which is equivalent to $E^*_{\alpha,\Lambda}(P)\leq E_\alpha(P)+C'\alpha^{-2(1+\sigma)}$.
\end{proof}

Given Theorem \ref{Theorem: Parabolic Lower Bound for the truncated Model}  we can now give a proof of Theorem \ref{Theorem: Parabolic Lower Bound}. 

\begin{proof}[Proof of Theorem \ref{Theorem: Parabolic Lower Bound}]
In the first step of the proof, we are going to verify Eq.~(\ref{Equation-Main}) for $|P|\leq \sqrt{2m}\alpha$. Due to the rotational symmetry, we can assume w.l.o.g. that $P=(P_1,0,0)$, and by Lemma \ref{Lemma-GroundStateEnergyComparison} we know that
\begin{align}
\nonumber
 E_\alpha(P) +C' \alpha^{-2(1+\sigma)}& \geq \inf\{E:(P_1,0,0,E)\in \sigma(\mathbb{P}_1,\mathbb{P}_2,\mathbb{P}_3,\mathbb{H}_\Lambda)\}\\
 \label{Equation-SpectrumLowerBound}
&\geq \inf\{E:(P_1,E)\in \sigma(\mathbb{P}_1,\mathbb{H}_\Lambda)\}.
\end{align}
Making use of the fact that the operators $\mathbb{P}'_1,\mathbb{H}_\Lambda', \mathbb{P}_1-\mathbb{P}'_1$ and $\mathbb{H}_\Lambda-\mathbb{H}_\Lambda'$ are pairwise commuting and that $\mathbb{P}'_1,\mathbb{H}_\Lambda'$ and $\mathbb{P}_1-\mathbb{P}'_1,\mathbb{H}_\Lambda-\mathbb{H}_\Lambda'$ act on different factors in the tensor product $L^2\! \bigg(\mathbb{R}^3,\mathcal{F}\Big(\mathrm{ran} \chi\big(|\nabla|\leq \Lambda\big)\Big)\bigg)\otimes \mathcal{F}\Big(\mathrm{ran} \chi\big(|\nabla|> \Lambda\big)\Big)$, their joint spectrum is well-defined and satisfies $\sigma\left(\mathbb{P}'_1,\mathbb{H}_\Lambda', \mathbb{P}_1-\mathbb{P}'_1,\mathbb{H}_\Lambda-\mathbb{H}_\Lambda'\right)=\sigma(\mathbb{P}_1',\mathbb{H}_\Lambda')\times \sigma(\mathbb{P}_1-\mathbb{P}_1',\mathbb{H}_\Lambda-\mathbb{H}_\Lambda')$. Hence we can rewrite the right hand side of Eq.~(\ref{Equation-SpectrumLowerBound}) as
\begin{align*}
\inf_{P_1'+\widetilde{P}_1=P_1}\Big\{E'+\widetilde{E}:(P_1',E')\in \sigma(\mathbb{P}_1',\mathbb{H}_\Lambda'),\ (\widetilde{P}_1,  \widetilde{E})\in \sigma(\mathbb{P}_1-\mathbb{P}_1',\mathbb{H}_\Lambda-\mathbb{H}_\Lambda')\Big\}.
\end{align*}
In order to verify that $E'+\widetilde{E}$ is bounded from below by the right hand side of Eq.~(\ref{Equation-Main}) for a suitable $w>0$ and $|P_1|\leq \sqrt{2m}\alpha$, let us first consider the case $\widetilde{E}\geq \alpha^{-2}$. Since $E'\in \sigma(\mathbb{H}_{\Lambda}')$, we have $E'\geq \inf \sigma(\mathbb{H}_\Lambda')\geq \inf \sigma(\mathbb{H})=E_\alpha$ and therefore
\begin{align*}
E'+\widetilde{E}\geq E_\alpha+\alpha^{-2}\geq e^\mathrm{Pek}-\frac{1}{2\alpha^2}\mathrm{Tr}\left[1-\sqrt{H^\mathrm{Pek}}\, \right]+\alpha^{-2}-\alpha^{-(2+w')}
\end{align*}
for a suitable $w'>0$, where we have used \cite[Theorem 1.1]{BS1}. Regarding the other case $\widetilde{E}<\alpha^{-2}$, note that we have $(\widetilde{P}_1,\widetilde{E})\in \sigma(\mathbb{P}_1-\mathbb{P}_1',\mathbb{H}_\Lambda-\mathbb{H}_\Lambda')= \{(0,0)\}\cup \bigcup_{\ell=1}^\infty \mathbb{R}\times \{\frac{\ell}{\alpha^2}\}$, and therefore $\widetilde{E}=0$ and $\widetilde{P}_1=0$. Hence $|P_1'|=|P_1|\leq \sqrt{2m}\alpha$ and consequently
\begin{align*}
E'+\widetilde{E}=E'&  \geq e^\mathrm{Pek}-\frac{1}{2\alpha^2}\mathrm{Tr}\left[1-\sqrt{H^\mathrm{Pek}}\, \right]+\frac{|P'_1|^2}{2\alpha^4 m}-\alpha^{-(2+w)}\\
&=e^\mathrm{Pek}-\frac{1}{2\alpha^2}\mathrm{Tr}\left[1-\sqrt{H^\mathrm{Pek}}\, \right]+\frac{|P_1|^2}{2\alpha^4 m}-\alpha^{-(2+w)},
\end{align*}
where we have used $(P_1',E')\in \sigma(\mathbb{P}_1',\mathbb{H}_\Lambda')$ together with Eq.~(\ref{Equation-EnergyMomentumRelation}). This concludes the proof of Eq.~(\ref{Equation-Main}) for $|P|\leq \sqrt{2m}\alpha$.

In order to verify Eq.~(\ref{Equation-Main}) for $|P|> \sqrt{2m}\alpha$, we are going to use the fact that $P\mapsto E_\alpha(P)$ is a monotone radial function, as recently shown in  \cite{Po}, and consequently $E_\alpha(P)\geq E_\alpha\! \left(\sqrt{2m}\frac{P}{|P|}\right)$ for $|P|\geq \sqrt{2m}\alpha$. This reduces the problem to the previous case, 
and hence concludes the proof of Theorem \ref{Theorem: Parabolic Lower Bound}.
\end{proof}

\section{Construction of a Condensate}
\label{Section: Construction of a Condensate}
This section is devoted to the construction of approximate $p$ ground states $\Psi_\alpha$ satisfying complete condensation in $\varphi^\mathrm{Pek}$, which we will utilize in order to prove Theorem \ref{Theorem: Parabolic Lower Bound for the truncated Model} in Section \ref{Section: Proof of Theorem}. In this context, we call $\Psi_\alpha$ an approximate $p$ ground state in case $\braket{\Psi_\alpha|\mathbb{H}_\Lambda|\Psi_\alpha}=E_{\alpha,\Lambda}\! \left(\alpha^2 p\right)+O_{\alpha\rightarrow \infty}\left(\alpha^{-(2+w)}\right)$, where $E_{\alpha,\Lambda}\! \left(\alpha^2 p\right)$ and $\mathbb{H}_\Lambda$ are defined in, respectively above, Theorem \ref{Theorem: Parabolic Lower Bound for the truncated Model}, and $\big\langle \Psi_\alpha\big|\left(\Upsilon_\Lambda-p\right)^2\big|\Psi_\alpha\big\rangle\lesssim \alpha^{-(2+w)}$, with $w>0$, where we define the (rescaled and truncated) phonon momentum operator 
\begin{align*}
\Upsilon_\Lambda:=\int \! \chi^1\! \left(\Lambda^{-1}|k_1|\leq 2\right)\! k_1\,  a_k^\dagger a_k\mathrm{d}k\,.
\end{align*}
Similarly to $\mathbb{H}_\Lambda$, it also depends on $\alpha$ due to the rescaled canonical commutation relations $[a(f),a^\dagger(g)]=\alpha^{-2}\braket{g|f}$ but we suppress the $\alpha$ dependence for the sake of readability. Here and in the following, we write $X\lesssim Y$ in case there exist constants $C,\alpha_0>0$ such that $X\leq C\, Y$ for all $\alpha\geq \alpha_0$. It is clear that there exist states $\Psi_\alpha$ that satisfy both $\braket{\Psi_\alpha|\mathbb{H}_\Lambda|\Psi_\alpha}-E_{\alpha,\Lambda}\! \left(\alpha^2 p\right)\lesssim \alpha^{-(2+w)}$ and $\big\langle \Psi_\alpha\big|\left(\alpha^{-2}\, \mathbb{P}_\Lambda-p\right)^2\big|\Psi_\alpha\big\rangle\lesssim \alpha^{-(2+w)}$, since $\left(p,E_{\alpha,\Lambda}(\alpha^2 p)\right)$ is a point in the joint spectrum of $\left(\alpha^{-2}\, \mathbb{P}_\Lambda,\mathbb{H}_\Lambda\right)$. As part of the subsequent Lemma \ref{Lemma-InitialState} we are going to show that the contribution of $\frac{1}{i\alpha^2}\nabla_{x_1}$ in $\alpha^{-2}\, \mathbb{P}_\Lambda=\frac{1}{i\alpha^2}\nabla_{x_1}+\Upsilon_\Lambda$ is negligibly small, i.e., we shall show that it does not matter whether one uses $\Upsilon_\Lambda$ or $\alpha^{-2}\, \mathbb{P}_\Lambda$ in the definition of approximate ground states. In particular, this will imply the existence of approximate $p$ ground states. We will choose $\Psi_\alpha$ such that $\mathrm{supp}\left(\Psi_\alpha\right)\subseteq B_{L}(0)$ for a suitable $L$, where we define the support using the identification $L^2\!\left(\mathbb{R}^3\right)\otimes \mathcal{F}\left(L^2\!\left(\mathbb{R}^3\right)\right)\cong L^2\!\left(\mathbb{R}^3,\mathcal{F}\left(L^2\!\left(\mathbb{R}^3\right)\right)\right)$ in order to represent elements $\Psi\in L^2\!\left(\mathbb{R}^3\right)\otimes \mathcal{F}\left(L^2\!\left(\mathbb{R}^3\right)\right)$ as functions $x\mapsto \Psi(x)$ with values in $\mathcal{F}\left(L^2\!\left(\mathbb{R}^3\right)\right)$, i.e. $\mathrm{supp}\left(\Psi\right)$ refers to the support of the electron.

In the rest of this paper, we will always assume that $\alpha\geq 1$.  Most of the results in this Section include $E_{\alpha,\Lambda}\! \left(\alpha^2 p\right)\leq E_\alpha+C |p|^2$ as an assumption for an arbitrary, but fixed, constant $C>0$, where $E_\alpha$ denotes the ground state energy of $\mathbb{H}$.  For the purpose of proving Theorem \ref{Theorem: Parabolic Lower Bound for the truncated Model} this is  not a  restriction, since we can always pick $C\geq \frac{1}{2m}$ and therefore $E_{\alpha,\Lambda}\! \left(\alpha^2 p\right)> E_\alpha+C|p|^2$ immediately implies the statement of Theorem \ref{Theorem: Parabolic Lower Bound for the truncated Model}
\begin{align*}
E_{\alpha,\Lambda}\! \left(\alpha^2 p\right)> E_\alpha+C|p|^2\geq e^\mathrm{Pek}-\frac{1}{2\alpha^2}\mathrm{Tr}\left[1-\sqrt{H^\mathrm{Pek}}\, \right]+\frac{|p|^2}{2m}-\alpha^{-(2+s)},
\end{align*}
where we used $E_\alpha\geq e^\mathrm{Pek}-\frac{1}{2\alpha^2}\mathrm{Tr}\left[1-\sqrt{H^\mathrm{Pek}}\, \right]-\alpha^{-(2+s)}$ by \cite[Theorem 1.1]{BS1}.

\begin{lem}
\label{Lemma-InitialState}
Given $0<\sigma<\frac{1}{4}$, let $\Lambda=\alpha^{\frac{4}{5}(1+\sigma)}$ and $L=\alpha^{1+\sigma}$, and assume $p$ satisfies $|p|\leq \frac{C}{\alpha}$ and $E_{\alpha,\Lambda}\! \left(\alpha^2 p\right)\leq E_\alpha+C|p|^2$ for a given $C>0$, where $E_\alpha$ is the ground state energy of $\mathbb{H}$. Then there exist states $\Psi_\alpha^\bullet$ satisfying $\big\langle \Psi^\bullet_\alpha\big|\left(\Upsilon_\Lambda-p\right)^2\big|\Psi_\alpha^\bullet\big\rangle\lesssim \alpha^{2\sigma-4}$ and $\braket{\Psi_\alpha^\bullet|\mathbb{H}_\Lambda|\Psi_\alpha^\bullet}-E_{\alpha,\Lambda}\! \left(\alpha^2 p\right)\lesssim \alpha^{-2(1+\sigma)}$, as well as $\mathrm{supp}\left(\Psi_\alpha^\bullet\right)\subseteq B_{L}(0)$.
\end{lem}
\begin{proof}
Since $\left(p,E_{\alpha,\Lambda}\! \left(\alpha^2 p\right)\right)$ is an element of the joint spectrum $\sigma\left(\frac{1}{i\alpha^2}\nabla_{x_1}+\Upsilon_\Lambda,\mathbb{H}_\Lambda\right)$, there exist states $\Psi^0_\alpha$ satisfying $\big\langle \Psi^0_\alpha\big|\left(\frac{1}{i\alpha^2}\nabla_{x_1}+\Upsilon_\Lambda-p\right)^2\big|\Psi^0_\alpha\big\rangle\leq \alpha^{-4}$ and
\begin{align}
\label{Equation-AdHockEnergy}
\braket{\Psi^0_\alpha|\mathbb{H}_\Lambda|\Psi^0_\alpha}\leq E_{\alpha,\Lambda}\! \left(\alpha^2 p\right)+\frac{1}{2}\alpha^{-2(1+\sigma)}.
\end{align}
From \cite[Lemma 2.4]{BS1} we know that $\braket{\Psi^0_\alpha|-\Delta_x|\Psi^0_\alpha}\leq 2\braket{\Psi^0_\alpha|\mathbb{H}_\Lambda|\Psi^0_\alpha}+d$ for a suitable constant $d>0$, which implies that $\braket{\Psi^0_\alpha|-\Delta_x|\Psi^0_\alpha}\lesssim 1$ due to Eq.~(\ref{Equation-AdHockEnergy}) and our assumption $E_{\alpha,\Lambda}\! \left(\alpha^2 p\right)\leq E_\alpha+C|p|^2\leq C|p|^2\leq \frac{C^3}{\alpha^2}$, and hence 
\begin{align}
\label{Equation: Phonon Momentum}
\big\langle \Psi^0_\alpha\big|\! \left(\Upsilon_\Lambda-p\right)^2\! \big|\Psi^0_\alpha\big\rangle \leq  2\bigg\langle \Psi^0_\alpha\bigg|\! \left(\frac{1}{i\alpha^2}\nabla_{x_1}+\Upsilon_\Lambda-p\right)^2\! \bigg|\Psi^0_\alpha\bigg\rangle\! -\! 2\alpha^{-4}\big\langle \Psi^0_\alpha\big|\Delta_x\big|\Psi^0_\alpha\big\rangle \leq  c \alpha^{-4}
\end{align}
for a suitable $c>0$. 

Let $\eta:\mathbb{R}^3\longrightarrow [0,\infty)$ be a smooth function that is supported on $B_1(0)$ and satisfies $\int \eta^2=1$. With this at hand we define $\Psi_y(x):=L^{-\frac{3}{2}}\eta\left(L^{-1}(x-y)\right)\Psi^0_\alpha(x)$ and $Z_y:=\|\Psi_y\|$, as well as the set $S\subseteq \mathbb{R}^3$ containing all $y$ satisfying $\braket{\Psi_y|\mathbb{H}_\Lambda|\Psi_y}> Z_y^2\left(E_{\alpha,\Lambda}\! \left(\alpha^2 p\right)+\left(1+\|\nabla \eta\|^2\right)\alpha^{-2(1+\sigma)}\right)$. Making use of the IMS identity we obtain
\begin{align*}
&\braket{\Psi^0_\alpha|\mathbb{H}_\Lambda|\Psi^0_\alpha}=\int \braket{\Psi_y|\mathbb{H}_\Lambda|\Psi_y}\mathrm{d}y-L^{-2}\|\nabla \eta\|^2\\
&\ \ \geq \int_S Z_y^2\mathrm{d}y\left(E_{\alpha,\Lambda}\! \left(\alpha^2 p\right)+\left(1+\|\nabla \eta\|^2\right)\alpha^{-2(1+\sigma)}\right)+\left(1-\int_S Z_y^2\, \mathrm{d}y\right)E_\alpha-L^{-2}\|\nabla \eta\|^2,
\end{align*}
where we have used $\braket{\Psi_y|\mathbb{H}_\Lambda|\Psi_y}\geq E_\alpha$ and $\int Z_y^2\, \mathrm{d}y=1$. Using Eq.~(\ref{Equation-AdHockEnergy}) and $L^{-2}=\alpha^{-2(1+\sigma)}$ therefore yields
\begin{align*}
\Big(E_{\alpha,\Lambda}\! \left(\alpha^2 p\right)-E_\alpha+ & \left(1+\|\nabla \eta\|^2\right)\alpha^{-2(1+\sigma)}\Big)\int_SZ_y^2\, \mathrm{d}y \\
&\leq E_{\alpha,\Lambda}\! \left(\alpha^2 p\right)-E_\alpha+\left(\frac{1}{2}+\|\nabla \eta\|^2\right)\alpha^{-2(1+\sigma)},
\end{align*}
and consequently $\int_S Z_y^2\, \mathrm{d}y\leq 1-\gamma_\alpha$ with $\gamma_\alpha:=\frac{1}{2}\frac{\alpha^{-2(1+\sigma)}}{E_{\alpha,\Lambda}\! \left(\alpha^2 p\right)-E_\alpha+\left(1+\|\nabla \eta\|^2\right)\alpha^{-2(1+\sigma)}}$. Let us further define $S'\subseteq \mathbb{R}^3$ as the set of all $y$ satisfying $\big\langle \Psi_y\big|\left(\Upsilon_\Lambda-p\right)^2\big|\Psi_y\big\rangle> Z_y^2\frac{2c}{\gamma_\alpha}\alpha^{-4}$. Clearly we have, using Eq.~(\ref{Equation: Phonon Momentum}),
\begin{align*}
\frac{2c}{\gamma_\alpha}\alpha^{-4} \int_{S'} Z_y^2\, \mathrm{d}y\leq \int \big\langle \Psi_y\big|\left(\Upsilon_\Lambda-p\right)^2\big|\Psi_y\big\rangle\, \mathrm{d}y=\big\langle \Psi^0_\alpha\big|\left(\Upsilon_\Lambda-p\right)^2\big|\Psi^0_\alpha\big\rangle\leq c\alpha^{-4},
\end{align*}
and hence $\int_{S'} Z_y^2\, \mathrm{d}y\leq \frac{\gamma_\alpha}{2}$. Consequently $\int_{S\cup S'}Z_y^2\, \mathrm{d}y\leq \int_{S}Z_y^2\, \mathrm{d}y+\int_{S'}Z_y^2\, \mathrm{d}y\leq 1-\frac{\gamma_\alpha}{2}<1$. Since $\int  Z_y^2\, \mathrm{d}y=1$, this means in particular that there exists a $y\notin S\cup S'$ with $Z_y>0$, i.e. $\Psi_\alpha^\bullet:=Z_y^{-1}\Psi_y$ satisfies $\braket{\Psi_\alpha^\bullet|\mathbb{H}_\Lambda|\Psi_\alpha^\bullet}\leq E_{\alpha,\Lambda}\! \left(\alpha^2 p\right)+\left(1+\|\nabla \eta\|^2\right)\alpha^{-2(1+\sigma)}$ and $\big\langle \Psi_\alpha^\bullet\big|\left(\Upsilon_\Lambda-p\right)^2\big|\Psi_\alpha^\bullet\big\rangle\leq \frac{2c}{\gamma_\alpha}\alpha^{-4}\lesssim \alpha^{2\sigma-4}$, where we have used $E_{\alpha,\Lambda}\! \left(\alpha^2 p\right)-E_\alpha\lesssim |p|^2\lesssim \alpha^{-2}$ in the last estimate. Moreover, we clearly have $\mathrm{supp}\left(\Psi_\alpha^\bullet\right)\subseteq B_{L}(y)$. By the translation invariance of $\mathbb{H}_\Lambda$ and $\Upsilon_\Lambda$, we can assume w.l.o.g. that $y=0$, which concludes the proof.
\end{proof}

In the following Lemmas~\ref{Lemma-SecondState} and~\ref{Lemma-ThirdState}, we will use localization methods in order to construct approximate $p$ ground states with useful additional properties, which we will use in Lemma \ref{Lemma-Existence of a condensate}, together with an additional localization procedure, in order to show the existence of approximate $p$ ground states satisfying complete condensation. In Theorem \ref{Theorem: Complete Condensation} we will then apply a final localization step in order to obtain complete condensation in a stronger sense, following the argument in \cite{LNSS}.

In order to formulate our various localization results, we follow \cite{BS1} and define for a function $F:\mathcal{M}\left(\mathbb{R}^3\right)\longrightarrow \mathbb{R}$, where $\mathcal{M}\left(\mathbb{R}^3\right)$ is the set of all finite (Borel) measures on $\mathbb{R}^3$, the operator $\widehat{F}$ acting on $\mathcal{F}\left(L^2(\mathbb{R}^3)\right)=\bigoplus\limits_{n=0}^\infty L^2_{\mathrm{sym}}\! (\mathbb{R}^{3\times n})$ as $\widehat{F}\bigoplus\limits_{n=0}^\infty \Psi_n:=\bigoplus\limits_{n=0}^\infty \Psi_n^*$ with $\Psi_n^*(x^1,\dots ,x^n):=F^n(x^1,\dots ,x^n)\Psi_n(x^1,\dots ,x^n)$, where
\begin{align}
\label{Equation: Definition F_n}
F^n(x^1,\dots ,x^n):=F\left(\alpha^{-2}\sum_{k=1}^n\delta_{x^k}\right)
\end{align}
and $\widehat{F}_0:=F(0)$, i.e. $\widehat{F}$ acts component-wise on $\bigoplus\limits_{n=0}^\infty L^2_{\mathrm{sym}}\! (\mathbb{R}^{3\times n})$ by multiplication with the real-valued function $(x^1,\dots ,x^n)\mapsto F\left(\alpha^{-2}\sum_{k=1}^n\delta_{x^k}\right)$.

With this notation at hand, we define for given positive $c_-,c_+$ and $\epsilon'$ the function $F_*(\rho):=\chi^{\epsilon'}\left(c_-+\epsilon'\leq \int \mathrm{d}\rho\leq c_+-\epsilon'\right)$ and the states
\begin{align}
\label{Equation:First Localization}
\Psi_\alpha':=Z_{\alpha}^{-1}\widehat{F}_*\Psi^\bullet_\alpha,
\end{align}
with normalization constants $Z_{\alpha}:=\|\widehat{F}_*\Psi^\bullet_\alpha\|$, where $\Psi^\bullet_\alpha$ is the sequence constructed in Lemma \ref{Lemma-InitialState}. Since $\mathcal{N}=\widehat{G}$ with $G(\rho):=\int \mathrm{d}\rho$, it is clear that the states $\Psi_\alpha'$ are localized to a region where the (scaled) number operator $\mathcal{N}$ is between $c_-$ and $c_+$, i.e. $\chi\left(c_-\leq \mathcal{N}\leq c_+\right)\Psi_\alpha'=\Psi_\alpha'$. The following Lemma \ref{Lemma-SecondState} quantifies the energy and momentum error of this localization procedure. The subsequent results in Lemmas~\ref{Lemma-SecondState},~\ref{Lemma-ThirdState} and~\ref{Lemma-Existence of a condensate} as well as Theorem \ref{Theorem: Complete Condensation}, which quantify the energy and momentum error of specific localization procedures, are generalizations of the corresponding results in \cite{BS1}, where only the energy cost of such localization procedures is discussed. In the following we will usually refer to the respective results in \cite{BS1} when it comes to quantifying the energy error, and only discuss the localization error of the momentum operator $\Upsilon_\Lambda$.
\begin{lem}
\label{Lemma-SecondState}
Given $0<\sigma<\frac{1}{4}$, let $\Lambda=\alpha^{\frac{4}{5}(1+\sigma)}$ and $L=\alpha^{1+\sigma}$, and assume $p$ satisfies $|p|\leq \frac{C}{\alpha}$ and $E_{\alpha,\Lambda}\! \left(\alpha^2 p\right)\leq E_\alpha+C|p|^2$ for a given $C>0$. Then there exist constants $c_-,c_+$ and $\epsilon'$, such that the states $\Psi_\alpha'$ defined in Eq.~(\ref{Equation:First Localization}) satisfy $\big\langle \Psi_\alpha'\big|\left(\Upsilon_\Lambda-p\right)^2\big|\Psi_\alpha'\big\rangle\lesssim \alpha^{2\sigma-4}$ and $\braket{\Psi_\alpha'|\mathbb{H}_\Lambda|\Psi_\alpha'}-E_{\alpha,\Lambda}\! \left(\alpha^2 p\right)\lesssim \alpha^{-2(1+\sigma)}$.
\end{lem}
\begin{proof}
 By our assumptions we clearly have $\widetilde{E}_\alpha-E_\alpha\lesssim \alpha^{-\frac{4}{29}}$ with $\widetilde{E}_\alpha:=\braket{\Psi^\bullet_\alpha|\mathbb{H}_\Lambda|\Psi^\bullet_\alpha}$, and therefore we can apply  \cite[Lemma 3.4]{BS1}, which tells us that we can choose $c_-,c_+$ and $\epsilon'$, such that $\braket{\Psi_\alpha'|\mathbb{H}_\Lambda|\Psi_\alpha'}-E_{\alpha,\Lambda}\! \left(\alpha^2 p\right)\lesssim \alpha^{-2(1+\sigma)}$, and furthermore $Z_\alpha\underset{\alpha \rightarrow \infty}{\longrightarrow}1$. Since $\widehat{F}_*$ commutes with $\Upsilon_\Lambda$, we obtain with $\widetilde{\Psi}_\alpha:=\sqrt{\frac{1-\widehat{F}_*^2}{1-Z_\alpha^2}}\Psi_\alpha^\bullet$
\begin{align*}
Z_\alpha^2\big\langle \Psi_\alpha'\big|\left(\Upsilon_\Lambda-p\right)^2\big|\Psi_\alpha'\big\rangle+(1-Z_\alpha^2)\big\langle \widetilde{\Psi}_\alpha\big|\left(\Upsilon_\Lambda-p\right)^2\big|\widetilde{\Psi}_\alpha\big\rangle=\big\langle \Psi_\alpha^\bullet\big|\left(\Upsilon_\Lambda-p\right)^2\big|\Psi_\alpha^\bullet\big\rangle
\end{align*}
Hence $\big\langle \Psi_\alpha'\big|\left(\Upsilon_\Lambda-p\right)^2\big|\Psi_\alpha'\big\rangle\leq Z_\alpha^{-2}\big\langle \Psi_\alpha^\bullet\big|\left(\Upsilon_\Lambda-p\right)^2\big|\Psi_\alpha^\bullet\big\rangle\lesssim \alpha^{2\sigma-4}$.
\end{proof}

When it comes to localizations with respect to more complicated functions $F$ compared to the one used in Eq.~(\ref{Equation:First Localization}), we first need to introduce some tools in order to quantify the localization error of the momentum operator. Given a function $F:\mathcal{M}\left(\mathbb{R}^3\right)\longrightarrow \mathbb{R}$, $\Omega\subseteq \mathcal{M}\left(\mathbb{R}^3\right)$ and $\lambda> 0$, let us define
\begin{align}
\label{Equation-DefinitionNorm}
\|F\|_{\Omega,\lambda}^2:=\sup_{1\leq n\leq \lambda \alpha^2}\sup_{x\in \Omega_n}\left\|(F^{n,\bar{x}})'\right\|^2=\sup_{1\leq n\leq \lambda \alpha^2}\sup_{x\in \Omega_n}\int_\mathbb{R} \left|\frac{\mathrm{d}}{\mathrm{d}t}F^n(t,\bar{x})\right|^2\! \mathrm{d}t,
\end{align}
where $x=(x^1,\dots,x^n)\in \mathbb{R}^{3\times n}$ with $x^k=(x_1^k,x_2^k,x_3^k)$ and $\bar{x}:=(x^1_2,x^1_3,x^2,\dots,x^n)\in \mathbb{R}^{3\times n-1}$, i.e. we define $\bar{x}$ such that $x=(x^1_1,\bar{x})$,  $\Omega_n$ is the set of all $x$ such that $\alpha^{-2}\sum_{j=1}^n\delta_{x^j}\in \Omega$ and $F^{n,y}:\mathbb{R}\longrightarrow \mathbb{R}$ is defined as $F^{n,y}(t):=F^n(t,y)$ for $y\in  \mathbb{R}^{3\times n-1}$, where $F^n$ is as in Eq.~(\ref{Equation: Definition F_n}).

\begin{lem}
\label{Lemma-IMS}
Given $\lambda>0$, there exists a constant $T>0$ such that we have for all quadratic partitions of unity $\mathcal{P}=\{F_j:\mathcal{M}\left(\mathbb{R}^3\right)\longrightarrow \mathbb{R}:j\in J\}$, i.e. families of functions satisfying $0\leq F_j\leq 1$ and $\sum_{j\in J}F_j^2=1$, $\Lambda>0$, $|p|\leq \Lambda$, $\Omega\subseteq \mathcal{M}\left(\mathbb{R}^3\right)$ and states $\Psi$ satisfying $\chi\left(\mathcal{N}\leq \lambda\right)\Psi=\Psi$ and $\widehat{\mathds{1}_\Omega}\Psi=\Psi$
\begin{align*}
\left|\sum_{j\in J}\braket{\Psi_j|\left(\Upsilon_\Lambda-p\right)^2|\Psi_j}-\braket{\Psi|\left(\Upsilon_\Lambda-p\right)^2|\Psi}\right|\leq T \Lambda\sum_{j\in J}\|F_j\|^2_{\Omega,\lambda},
\end{align*}
where we define $\Psi_j:=\widehat{F}_j\Psi$.
\end{lem}
\begin{proof}
Using the IMS identity we can write
\begin{align*}
\sum_{j\in J}\braket{\Psi_j|\left(\Upsilon_\Lambda-p\right)^2|\Psi_j}-\braket{\Psi|\left(\Upsilon_\Lambda-p\right)^2|\Psi}=-\frac{1}{2}\sum_{j\in J}\big\langle \Psi\big| \left[\left[\left(\Upsilon_\Lambda-p\right)^2,\widehat{F}_j\right],\widehat{F}_j\right]\big|\Psi \big\rangle.
\end{align*}
Hence it suffices to  show that $\pm\big\langle \Psi\big| \left[\left[\left(\Upsilon_\Lambda-p\right)^2,\widehat{F}\right],\widehat{F}\right]\big|\Psi \big\rangle\lesssim \Lambda\|F\|_{\Omega,\lambda}^2 $ for any bounded $F:\mathcal{M}\left(\mathbb{R}^3\right)\longrightarrow \mathbb{R}$ and state satisfying $\chi\left(\mathcal{N}\leq \lambda\right)\Psi=\Psi$ and $\widehat{\mathds{1}_\Omega}\Psi=\Psi$. Let us start by estimating
\begin{align*}
&\pm \left[\left[\left(\Upsilon_\Lambda-p\right)^2,\widehat{F}\right],\widehat{F}\right]=\pm 2\left[\Upsilon_\Lambda,\widehat{F}\right]^2 \pm \Big\{\Upsilon_\Lambda-p,\left[\left[\Upsilon_\Lambda,\widehat{F}\right],\widehat{F}\right]\Big\}\\
&\ \ \ \ \leq -2\left[\Upsilon_\Lambda,\widehat{F}\right]^2+\frac{\|F\|_{\Omega,\lambda}^2}{\Lambda}\left(\Upsilon_\Lambda-p\right)^2+\frac{\Lambda}{\|F\|_{\Omega,\lambda}^2} \left[\left[\Upsilon_\Lambda,\widehat{F}\right],\widehat{F}\right]^2,
\end{align*}
where $\{A,B\}:=AB+BA$. By the definition of $\Upsilon_\Lambda$ it is clear that $\frac{\|F\|_{\Omega,\lambda}^2}{\Lambda}\left(\Upsilon_\Lambda-p\right)^2\lesssim \Lambda\|F\|_{\Omega,\lambda}^2\left(\mathcal{N}+1\right)^2$ for $|p|\leq \Lambda$, and consequently $\pm\big\langle \Psi\big| \frac{\|F\|_{\Omega,\lambda}^2}{\Lambda}\left(\Upsilon_\Lambda-p\right)^2\big|\Psi \big\rangle\lesssim \Lambda\|F\|_{\Omega,\lambda}^2 $. Using that $\Psi$ is a function with values in $\mathcal{F}_{\leq \lambda \alpha^2}\left(L^2(\mathbb{R}^3)\right):=\bigoplus\limits_{n\leq \lambda \alpha^2} L^2_{\mathrm{sym}}\! (\mathbb{R}^{3\times n})$, we are going to represent it as $\Psi=\bigoplus_{n\leq \lambda \alpha^2}\Psi_n$ where $\Psi_n(z,x^1,\dots ,x^{n})$ is a function of the electron variable $z$ and the $n$ phonon coordinates $x^j\in \mathbb{R}^3$ satisfying $\Psi_n(z,x^1,\dots ,x^{n})=0$ for all $(x^1,\dots,x^n)\notin\Omega_n$. In order to simplify the notation, we will suppress the dependence on the electron variable $z$. We have $\left[\Upsilon_\Lambda,\widehat{F}\right]\Psi=\bigoplus_{1\leq n\leq \lambda \alpha^2} \alpha^{-2}n\Psi^*_n$ with $\Psi_n^*:=\frac{1}{n}\sum_{j=1}^n\left[g\left(\frac{1}{i}\nabla_{x^j_1}\right),F^n\right]\Psi_n$, where $g(k):=\! \chi^1\! \left(\Lambda^{-1}|k|\leq 2\right)\!k$ for $k\in \mathbb{R}$. Hence
\begin{align*}
\Big\langle \Psi \Big|-\left[\Upsilon_\Lambda,\widehat{F}\right]^2\Big|\Psi \Big\rangle=\left\|\left[\Upsilon_\Lambda,\widehat{F}\right]\Psi\right\|^2=\sum_{1\leq n\leq \lambda \alpha^2}\alpha^{-4}n^2\|\Psi^*_n\|^2\leq \lambda^2 \sum_{1\leq n\leq \lambda \alpha^2}\|\Psi^*_n\|^2,
\end{align*}
and $\|\Psi^*_n\|\leq \frac{1}{n}\sum_{j=1}^n\left\|\left[g\left(\frac{1}{i}\nabla_{x^j_1}\right),F^n\right]\Psi_n\right\|=\left\|\left[g\left(\frac{1}{i}\nabla_{x^1_1}\right),F^n\right]\Psi_n\right\|$, where we have used the permutation symmetry of $\Psi_n$. By Lemma \ref{Lemma-IMSforGradient} we know that
\begin{align*}
\left\|\left[g\left(\frac{1}{i}\nabla_{x^1_1}\right),F^n\right]\Psi_n\right\|   \! \leq  \!  \!  \sup_{x\in \mathrm{supp}(\Psi_n)}\left\|\! \left[g\left(\frac{1}{i}\frac{\mathrm{d}}{\mathrm{d}t}\right),F^{n,\bar{x}}\right]\! \right\|_\mathrm{op} \!  \! \|\Psi_n\| \! \lesssim \!  \sqrt{\Lambda} \! \sup_{x\in \Omega_n} \! \|(F^{n,\bar{x}})'\|\|\Psi_n\|,
\end{align*}
and therefore
\begin{align*}
\Big\langle \Psi \Big|-\left[\Upsilon_\Lambda,\widehat{F}\right]^2\Big|\Psi \Big\rangle\leq \lambda^2 \Lambda \sup_{1\leq n\leq \lambda \alpha^2,x\in \Omega_n}\| (F^{n,\bar{x}})'\|^2 \sum_{n\leq \lambda \alpha^2} \|\Psi_n\|^2=\lambda^2 \Lambda \|F\|_{\Omega,\lambda}^2.
\end{align*}
In order to estimate the expectation value of $\left[\left[\Upsilon_\Lambda,\widehat{F}\right],\widehat{F}\right]^2$ we proceed similarly, by writing $\left[\left[\Upsilon_\Lambda,\widehat{F}\right],\widehat{F}\right]\Psi=\bigoplus_{n\leq \lambda \alpha^2} \alpha^{-2}n\widetilde{\Psi}_n$ with $\widetilde{\Psi}_n=\frac{1}{n}\sum_{j=1}^n\left[\left[g\left(\frac{1}{i}\nabla_{x^j_1}\right),F^n\right],F^n\right]\Psi_n$, and estimating $\Big\langle \Psi \Big|\left[\left[\Upsilon_\Lambda,\widehat{F}\right],\widehat{F}\right]^2\Big|\Psi \Big\rangle\leq \lambda^2 \sum_{n\leq \lambda \alpha^2}\left\|\widetilde{\Psi}_n\right\|^2$ as well as
\begin{align*}
\left\|\widetilde{\Psi}_n\right\|\leq \sup_{x\in \mathrm{supp}(\Psi_n)}\left\|\left[\left[g\left(\frac{1}{i}\frac{\mathrm{d}}{\mathrm{dt}}\right),F^{n,\bar{x}}\right],F^{n,\bar{x}}\right]\right\|_\mathrm{op}\|\Psi_n\|\leq \sup_{x\in \Omega_n} \|(F^{n,\bar{x}})'\|^2 \|\Psi_n\|,
\end{align*}
where we have again applied Lemma \ref{Lemma-IMSforGradient}. This concludes the proof.
\end{proof}

With the subsequent localization step in Eq.~(\ref{Equation:Second Localization}), we want to restrict the state $\Psi'_\alpha$ to phonon density configurations $\rho$ which have a sharp concentration of their mass. To be precise, for given $R$ and $\epsilon,\delta>0$, let us define $K_R\left(\rho\right):=\iint \chi^{\epsilon}\left(R-\epsilon\leq |x-y|\right)\mathrm{d}\rho(x)\mathrm{d}\rho(y)$ as well as $F_R\left(\rho\right):=\chi^{\frac{\delta}{3}}\Big(K_R\left(\rho\right)\leq \frac{2\delta}{3}\Big)$ and
\begin{align}
\label{Equation:Second Localization}
\Psi''_\alpha:=Z_{R,\alpha}^{-1}\widehat{F}_R\Psi'_\alpha,
\end{align}
where $\Psi'_\alpha$ is as in Lemma \ref{Lemma-SecondState} and $Z_{R,\alpha}:=\|\widehat{F}_R\Psi'_\alpha\|$. Clearly $\widehat{\mathds{1}_\Omega}\Psi''_\alpha=\Psi''_\alpha$ where $\Omega$ is the set of all $\rho$ satisfying $\iint_{|x-y|\geq R} \mathrm{d}\rho(x)\mathrm{d}\rho(y)\leq \delta$. In the following Lemma \ref{Lemma-ThirdState} we are going to quantify the energy and momentum cost of this localization procedure.

\begin{lem}
\label{Lemma-ThirdState}
Given $0<\sigma<\frac{1}{4}$, let $\Lambda=\alpha^{\frac{4}{5}(1+\sigma)}$ and $L:=\alpha^{1+\sigma}$, and assume $p$ satisfies $|p|\leq \frac{C}{\alpha}$ and $E_{\alpha,\Lambda}\! \left(\alpha^2 p\right)\leq E_\alpha+C|p|^2$ for a given $C>0$. Then for any $\epsilon,\delta>0$, there exists a constant $R>0$, such that the states $\Psi_\alpha''$ defined in Eq.~(\ref{Equation:Second Localization}) satisfy $\big\langle \Psi_\alpha''\big|\left(\Upsilon_\Lambda-p\right)^2\big|\Psi_\alpha''\big\rangle\lesssim \alpha^{\frac{4}{5}\sigma-\frac{16}{5}}$ and $\braket{\Psi_\alpha''|\mathbb{H}_\Lambda|\Psi_\alpha''}-E_{\alpha,\Lambda}\! \left(\alpha^2 p\right)\lesssim \alpha^{-2(1+\sigma)}$.
\end{lem}
\begin{proof}
By the results in \cite[Lemma 3.5]{BS1}, there exists a constant $R>0$ such that $\braket{\Psi_\alpha''|\mathbb{H}_\Lambda|\Psi_\alpha''}-E_{\alpha,\Lambda}\! \left(\alpha^2 p\right)\lesssim \alpha^{-2(1+\sigma)}$ and $Z_{R,\alpha}\underset{\alpha\rightarrow \infty}{\longrightarrow}1$. Applying Lemma \ref{Lemma-IMS} yields
\begin{align}
\nonumber
&\braket{\widehat{F}_R\Psi_\alpha'|\left(\Upsilon_\Lambda-p\right)^2|\widehat{F}_R\Psi_\alpha'}\! +\! \braket{\widehat{G}_R\Psi_\alpha'|\left(\Upsilon_\Lambda-p\right)^2|\widehat{G}_R\Psi_\alpha'}\\
\label{Equation-Application of IMS}
& \ \ \ \ \ \ \  \lesssim  \! \alpha^{2\sigma-4}\! +\! \alpha^{\frac{4}{5}(1+\sigma)} \! \left(\|F_R\|^2_{\mathcal{M}\left(\mathbb{R}^3\right),c_+}\! +\! \|G_R\|^2_{\mathcal{M}\left(\mathbb{R}^3\right),c_+}\right)
\end{align}
with $G_R:=\sqrt{1-F_R^2}$, where we used $\braket{\Psi_\alpha'|\left(\Upsilon_\Lambda-p\right)^2|\Psi_\alpha'}\lesssim \alpha^{2\sigma-4}$ and $\chi\left(\mathcal{N}\leq c_+\right)\Psi_\alpha'=\Psi_\alpha'$. In order to estimate $\|F_R\|_{\mathcal{M}\left(\mathbb{R}^3\right),c_+}$, let us define the functions $g(s):=\chi^{\frac{\delta}{3}}\Big(s\leq \frac{2\delta}{3}\Big)$ and $h(s):=\chi^\epsilon\left(R-\epsilon\leq \sqrt{s}\right)$. Then $F_R^n(x)=g\left(\alpha^{-4}\sum_{i,j=1}^n h\left(|x^i-x^j|^2\right)\right)$ and therefore $F_R^{n,y}(t)=g\left(\alpha^{-4}\sum_{i=2}^n h\left((t-y_1^j)^2+\delta^i_y\right)+\mu_y\right)$ with $\delta^i_y:=\left(y_2^1-y_2^i\right)^2+\left(y_3^1-y_3^i\right)^2$ and $\mu_y:=\alpha^{-4}\sum_{i,j=2}^n h\left(|y^i-y^j|^2\right)$. Consequently
\begin{align*}
\|(F_R^{n,y})'\|\! \leq\!  4\alpha^{-4}\| g'\|_\infty\! \sum_{i=2}^n\! \sqrt{\int_{\mathbb{R}}\!  |t|^2\left|h'\left(t^2+\delta^i_y\right)\right|^2\mathrm{d}t}\! \leq \! 4\alpha^{-4}\| g'\|_\infty (n-1)\|h'\|_\infty\sqrt{\frac{2R^3}{3}},
\end{align*}
where we have used $\mathrm{supp}\left(h'\right)\subseteq [0,R^2)$ in the second inequality. Hence $\|F_R\|_{\mathcal{M}\left(\mathbb{R}^3\right),c_+}=\sup_{1\leq n\leq c_+ \alpha^2}\sup_{x\in \mathbb{R}^{3\times n}}\|(F_R^{n,\bar{x}})'\|\lesssim \alpha^{-2}$. Similarly we have $\|G_R\|_{\mathcal{M}\left(\mathbb{R}^3\right),c_+}\lesssim \alpha^{-2}$. In combination with Eq.~(\ref{Equation-Application of IMS}) we therefore obtain
\begin{align*}
\big\langle \Psi_\alpha''\big|\! \left(\Upsilon_\Lambda-p\right)^2\! \big|\Psi_\alpha''\big\rangle\lesssim Z_{R,\alpha}^{-2}\left(\alpha^{2\sigma-4}\! +\! \alpha^{\frac{4}{5}(1+\sigma)}\!  \left(\|F_R\|^2_{\mathcal{M}\left(\mathbb{R}^3\right),c_+}\! \! \! +\! \|G_R\|^2_{\mathcal{M}\left(\mathbb{R}^3\right),c_+}\right)\right)\! \lesssim \! \alpha^{\frac{4}{5}\sigma-\frac{16}{5}}.
\end{align*}
\end{proof}

Before we come to our next localization step in Lemma \ref{Lemma-Existence of a condensate}, we need to define the regularized median of a measure $\nu\in \mathcal{M}\left(\mathbb{R}\right)$, see also \cite[Definition 3.8]{BS1}, and derive a useful estimate for it in the subsequent Lemma \ref{Lemma-MedianEstimate}. In the following let $x^\kappa(\nu):=\sup\{t:\int_{-\infty}^t\mathrm{d}\nu\leq \kappa \int\mathrm{d}\nu\}$ denote the $\kappa$-quantile, where we use the convention that boundaries are included in the domain of integration $\int_a^b f\mathrm{d}\nu:=\int_{[a,b]}f\mathrm{d}\nu$, and let us define for $0<q<\frac{1}{2}$ and $\nu\neq 0$
\begin{align}
\label{Equation: DefinitionRegularizedMedian}
m_q(\nu):=\frac{1}{\int_{K_q(\nu)} \mathrm{d}\nu}\ \int_{K_q(\nu)} h\, \mathrm{d}\nu(h),
\end{align}
where $K_q(\nu):=[x^{\frac{1}{2}-q}(\nu),x^{\frac{1}{2}+q}(\nu)]$, and $m_q(0):=0$. Furthermore we will denote the marginal measures of $\rho\in \mathcal{M}\left(\mathbb{R}^3\right)$ as $\rho_i$, i.e. $\rho_i(A):=\rho\left([x_i\in A]\right)$, where $A\subseteq \mathbb{R}$ is measurable and $i\in \{1,2,3\}$.

\begin{lem}
\label{Lemma-MedianEstimate}
Let us define $\Omega_\mathrm{reg}$ as the set of all $\rho\in \mathcal{M}\left(\mathbb{R}^3\right)$ satisfying $\int_{x_i=t}\mathrm{d}\rho(x)\leq \alpha^{-2}$ for $t\in \mathbb{R}$ and $i\in \{1,2,3\}$, and $\Omega$ as the set of all $\rho\in \Omega_\mathrm{reg}$ satisfying $c\leq \int \mathrm{d}\rho$ and $\iint_{|x-y|\geq R}\mathrm{d}\rho(x)\mathrm{d}\rho(y)\leq \delta$ for given $R,c,\delta>0$. Furthermore let $q$ be a constant satisfying $q+\frac{\alpha^{-2}}{c}\leq \frac{1}{2}-\frac{\delta}{c^2}$. Then we have for any $n\geq 1$ and function of the form $F(\rho)=f\left(m_q(\rho_1)\right)$ the estimate
\begin{align}
\label{Equation-MedianEstimate}
\mathrm{sup}_{x\in \Omega_n} \left\|(F^{n,\bar{x}})'\right\|\leq \alpha^{-2}\frac{\|f'\|_\infty}{2qc}\sqrt{2R},
\end{align}
where $m_q$ is defined in Eq.~(\ref{Equation: DefinitionRegularizedMedian}) and $\Omega_n$ below Eq.~(\ref{Equation-DefinitionNorm}).
\end{lem}
\begin{proof}
Given $x\in \Omega_n$, let us define $\nu_t:=\alpha^{-2}\left(\delta_t+\sum_{j=2}^n \delta_{x^j_1}\right)$, which allows us to rewrite $F^{n,\bar{x}}(t)=f(m_q(\nu_t))$. Let us first compute the derivative $\frac{\mathrm{d}}{\mathrm{d}t} m_q(\nu_t)$ for $t\in \mathbb{R}\setminus \{x_1^2,\dots ,x_1^n\}$. For such $t$, there clearly exists an $\epsilon>0$ such that $(t-\epsilon,t+\epsilon)\subset \mathbb{R}\setminus \{x_1^2,\dots ,x_1^n\}$. It will be useful in the following that the set $Y:=\{x_1^2,\dots ,x_1^n\}\cap K_q(\nu_s)$ is independent of $s\in (t-\epsilon,t+\epsilon)$, with $K_q(\nu)$ being defined below Eq.~(\ref{Equation: DefinitionRegularizedMedian}). Furthermore we have for $s\in (t-\epsilon,t+\epsilon)$ that $s\in K_q(s)$ if and only if $t\in K_q(t)$. Therefore $\alpha^2\int_{K_q(v_s)}h\, \mathrm{d}\nu_s(h)=\sum_{h\in Y}h+s\mathds{1}_{K_q(s)}(s)=\sum_{h\in Y}h+s\mathds{1}_{K_q(t)}(t)$ and $\alpha^2 \int_{K_q(v_s)}\mathrm{d}\nu_s=|Y|+\mathds{1}_{K_q(s)}(s)=\alpha^2\int_{K_q(v_t)}\mathrm{d}\nu_t$ for $s\in (t-\epsilon,t+\epsilon)$, and consequently we obtain for $t\in \mathbb{R}\setminus \{x_1^2,\dots ,x_1^n\}$
\begin{align*}
\frac{\mathrm{d}}{\mathrm{d}t} m_q(\nu_t)=\alpha^{-2}\frac{\mathrm{d}}{\mathrm{d}s}\Big|_{s=t}\frac{\sum_{h\in Y}h+s\mathds{1}_{K_q(t)}(t)}{\int_{K_q(v_t)}\mathrm{d}\nu_t}=\alpha^{-2}\frac{\mathds{1}_{K_q(t)}(t)}{\int_{K_q(v_t)}\mathrm{d}\nu_t}.
\end{align*}
Note that due to our assumption $\rho\in \Omega_\mathrm{reg}$, $m_q(\nu_t)$ can be continuously extended from $\mathbb{R}\setminus \{x_1^2,\dots ,x_1^n\}$ to all of $\mathbb{R}$, and therefore $\frac{\mathrm{d}}{\mathrm{d}t} m_q(\nu_t)=\alpha^{-2}\frac{\mathds{1}_{K_q(t)}(t)}{\int_{K_q(v_t)}\mathrm{d}\nu_t}$ in the sense of distributions. Since $\int_{K_q(v_t)}\mathrm{d}\nu_t\geq 2qc$ we conclude $|(F^{n,\bar{x}})'(t)|\leq \alpha^{-2}\frac{\|f'\|_\infty}{2qc}\mathds{1}_{K_q(t)}(t)$ for almost every $t$. In order to obtain from this the upper bound on the $L^2\! \left(\mathbb{R}\right)$-norm in Eq.~(\ref{Equation-MedianEstimate}), we are going to verify that the support of $t\mapsto \mathds{1}_{K_q(t)}(t)$ is contained in an interval of the form $(\xi-R,\xi+R)$ for a suitable $\xi\in \mathbb{R}$. Let us start by verifying that
\begin{align}
\label{Equation-QuantileCompare}
x^{\kappa}\left(\nu_{t_1}\right)\geq x^{\kappa-\frac{\alpha^{-2}}{c}}\left(\nu_{t_2}\right)
\end{align}
for $0<\kappa<1$ and $t_1,t_2\in \mathbb{R}$. Note that any $y\in \mathbb{R}$ satisfying the inequality $\int_{-\infty}^y \mathrm{d}\nu_{t_2}\leq \left(\kappa-\frac{\alpha^{-2}}{c}\right)\int \mathrm{d}\nu_{t_2}$, also satisfies
\begin{align*}
\int_{-\infty}^y \mathrm{d}\nu_{t_1}\leq \alpha^{-2}+\int_{-\infty}^y \mathrm{d}\nu_{t_2}\leq \alpha^{-2}+\left(\kappa-\frac{\alpha^{-2}}{c}\right)\int \mathrm{d}\nu_{t_2}\leq \kappa \int \mathrm{d}\nu_{t_2}=\kappa \int \mathrm{d}\nu_{t_1},
\end{align*}
where we have used $\alpha^{-2}\leq \frac{\alpha^{-2}}{c}\int \mathrm{d}\nu_{t_2}$, and therefore $y\leq x^{\kappa}\left(\nu_{t_1}\right)$. Using that $x^{\kappa-\frac{\alpha^{-2}}{c}}(\nu_{t_2})$ is the supremum over all such $y$, we conclude with the desired Eq.~(\ref{Equation-QuantileCompare}). Furthermore observe that $\nu_{t_0}=\rho_1$ with $t_0:=x^1_1$ and $\rho:=\alpha^{-2}\sum_{j=1}^n \delta_{x^j}\in \Omega$, and therefore we know by  \cite[Lemma 3.9]{BS1} that there exists a $\xi\in \mathbb{R}$ such that $\xi-R\leq x^{\frac{1}{2}-q'}(\nu_{t_0})\leq x^{\frac{1}{2}+q'}(\nu_{t_0})\leq \xi+R$ for $q'\leq \frac{1}{2}-\frac{\delta}{c^2}$. By our assumptions, $q':=q+\frac{\alpha^{-2}}{c}$ satisfies this condition, and therefore we obtain using Eq.~(\ref{Equation-QuantileCompare}) with $t_1:=t$, $t_2:=t_0$ and $\kappa:=\frac{1}{2}-q$, respectively $t_1:=t_0$, $t_2:=t$ and $\kappa:=\frac{1}{2}+q+\frac{\alpha^{-2}}{c}$, that 
\begin{align*}
\xi-R\leq x^{\frac{1}{2}-q}(\nu_{t})\leq x^{\frac{1}{2}+q}(\nu_{t})\leq \xi+R
\end{align*}
for all $t\in \mathbb{R}$, and consequently $\mathds{1}_{K_q(t)}(t)=0$ for $|t-\xi|> R$. 
\end{proof}

\begin{lem}
\label{Lemma-Existence of a condensate}
Given $0<\sigma<\frac{1}{9}$ and $C>0$, let $\Lambda=\alpha^{\frac{4}{5}(1+\sigma)}$ and $L=\alpha^{1+\sigma}$, and assume $p$ satisfies $|p|\leq \frac{C}{\alpha}$ and $E_{\alpha,\Lambda}\! \left(\alpha^2 p\right)\leq E_\alpha+C|p|^2$ for a given $C>0$. Then there exist $r',c_+>0$ and states $\Psi'''_\alpha$ with $\big\langle \Psi'''_\alpha\big|\left(\Upsilon_\Lambda-p\right)^2\big|\Psi'''_\alpha\big\rangle\lesssim \alpha^{-(2+r')}$, $\braket{\Psi'''_\alpha|\mathbb{H}_\Lambda|\Psi'''_\alpha}-E_{\alpha,\Lambda}\! \left(\alpha^2 p\right)\lesssim \alpha^{-(2+r')}$, $\mathrm{supp}\left(\Psi'''_\alpha\right)\subseteq B_{4L}(0)$ and $\chi\left(\mathcal{N}\leq c_+\right)\Psi'''_\alpha=\Psi_\alpha'''$, such that
\begin{align}
\label{Equation-Weak Codensation}
 \Big\langle \Psi'''_\alpha\Big|W_{\varphi^\mathrm{Pek}}^{-1}\, \mathcal{N}W_{\varphi^\mathrm{Pek}}\Big| \Psi'''_\alpha\Big\rangle\lesssim \alpha^{-r'},
\end{align}
where $W_{\varphi^\mathrm{Pek}}$ is the Weyl operator corresponding to the Pekar minimizer $\varphi^\mathrm{Pek}$, characterized by $W_{\varphi^\mathrm{Pek}}^{-1} a(f)W_{\varphi^\mathrm{Pek}}=a(f)-\braket{f|\varphi^\mathrm{Pek}}$ for all $f\in L^2\!\left(\mathbb{R}^3\right)$.
\end{lem}
\begin{proof}
For $u>0$, let us define the functions $f_\ell(y):=\chi^{\frac{1}{2}}\left(\ell-\frac{1}{2}<\alpha^u y\leq \ell+\frac{1}{2}\right)$ for $\ell\in \mathbb{Z}$ satisfying $|\ell|\leq \frac{3}{2}\alpha^u L$, as well as $f_{-\infty}(y):=\chi^{\frac{1}{2}}\left(\alpha^u y\leq -\lfloor \frac{3}{2}\alpha^u L\rfloor-\frac{1}{2}\right)$ and $f_{\infty}(\rho):=\chi^{\frac{1}{2}}\left(\lfloor \frac{3}{2}\alpha^u L\rfloor+\frac{1}{2}<\alpha^u y\right)$. With these functions at hand we define for $i\in \{1,2,3\}$  and $v>0$ the partitions $\mathcal{P}_i:=\big\{F_{\ell,i}:\ell\in A\big\}$, where $F_{\ell,i}(\rho):=f_\ell\left(m_{\alpha^{-v}}(\rho_i)\right)$ and $A:=\{-\infty, -\lfloor \frac{3}{2}\alpha^u L\rfloor, -\lfloor \frac{3}{2}\alpha^u L\rfloor+1,\dots,\lfloor \frac{3}{2}\alpha^u L\rfloor,\infty\}\subseteq \mathbb{Z}\cup \{-\infty,\infty\}$, as well as $\mathcal{P}:=\big\{F_z:z\in A^3\big\}$ with $F_z:=F_{z_3,3}F_{z_2,2}F_{z_1,1}$. In the following let $\Psi_\alpha''$ be as in Lemma \ref{Lemma-ThirdState} with $\delta<\frac{c_-^2}{2}$ and let $\Omega_\mathrm{reg}$ and $\Omega$ be the sets from Lemma \ref{Lemma-MedianEstimate} with $\delta$ and $R$ as in Lemma \ref{Lemma-ThirdState}, $q:=\alpha^{-v}$ and $c:=c_-$. Due to the straightforward result \cite[Lemma 3.6]{BS1} we have $\widehat{\mathds{1}_{\Omega_\mathrm{reg}}}\Psi_\alpha''=\Psi_\alpha''$, and by the definition of $\Psi_\alpha''$ in Eq.~(\ref{Equation:Second Localization}) it is clear that we furthermore have $\widehat{\mathds{1}_\Omega}\Psi_\alpha''=\Psi_\alpha''$. Therefore we can apply Lemma \ref{Lemma-IMS} together with Eq.~(\ref{Equation-MedianEstimate}) in order to obtain
\begin{align*}
 & \sum_{{z_1}\in A} \! \Big\langle \widehat{F}_{z_1,1}\Psi_\alpha''\Big|\left(\Upsilon_\Lambda-p\right)^2 \Big|\widehat{F}_{z_1,1}\Psi_\alpha''\Big\rangle \!  \leq \!  \braket{\Psi_\alpha''|\left(\Upsilon_\Lambda-p\right)^2 |\Psi_\alpha''} \! + \! T\alpha^{\frac{4}{5}(1+\sigma)} \! \sum_{{z_1}\in A} \! \alpha^{-4}\frac{\left\|f_{z_1}'\right\|^2_\infty}{2\alpha^{-2v}c^2_-}R\\
& \ \ \ \ \ \ \ \ \lesssim \alpha^{\frac{4}{5}\sigma-\frac{16}{5}}+\alpha^{\frac{4}{5}\sigma-\frac{16}{5}+2v}\sup_{{z_1}\in A}\left\|f_{z_1}'\right\|^2_\infty\sum_{{z_1}\in A}1\lesssim \alpha^{\frac{9}{5}\sigma+2v+3u-\frac{1}{5}}\alpha^{-2}
\end{align*}
for all $\alpha$ large enough such that $\alpha^{-v}+\frac{\alpha^{-2}}{c_-}<\frac{1}{2}-\frac{\delta}{c_-^2}$, where we have used $\sup_{{z_1}\in A}\left\|f_{z_1}'\right\|\lesssim \alpha^{u}$, as well as $\sum_{{z_1}\in A}1\leq 3(\alpha^u L+1)\lesssim \alpha^{u+1+\sigma}$. Since the functions $F^{n}_{\ell,i}$ are independent of $x_1^1$ for $i\in \{2,3\}$, we  furthermore obtain
\begin{align*}
\Big\langle \widehat{F}_{z_1,1}\Psi_\alpha''\Big|\left(\Upsilon_\Lambda-p\right)^2 \Big|\widehat{F}_{z_1,1}\Psi_\alpha''\Big\rangle= \!  \!  \! \sum_{z_2,z_3\in A} \!  \! \Big\langle \widehat{F}_{z_3,3}\widehat{F}_{z_2,2}\widehat{F}_{z_1,1}\Psi_\alpha''\Big|\left(\Upsilon_\Lambda-p\right)^2 \Big|\widehat{F}_{z_3,3}\widehat{F}_{z_2,2}\widehat{F}_{z_1,1}\Psi_\alpha''\Big\rangle
\end{align*}
and therefore
\begin{align}
\label{Equation-IMSforMomentum}
\sum_{z\in A^3}Z_z^2 \braket{\Psi_z|\left(\Upsilon_\Lambda-p\right)^2|\Psi_z}\lesssim \alpha^{\frac{9}{5}\sigma+2v+3u-\frac{1}{5}}\alpha^{-2}
\end{align}
with $\Psi_z:=Z_z^{-1}\widehat{F}_z\Psi''_\alpha$ and $Z_z:=\left\|\widehat{F}_z\Psi''_\alpha\right\|$.

Regarding the localization error of the energy, we obtain by  \cite[Lemma 3.3]{BS1} and \cite[Lemma 3.10]{BS1} (see also the proof of \cite[Eq.~(3.22)]{BS1}) that
\begin{align}
\label{Equation-IMSforEnergy}
\sum_{z\in A^3}Z_z^2 \braket{\Psi_z|\mathbb{H}_\Lambda|\Psi_z}\leq \braket{\Psi''_\alpha|\mathbb{H}_\Lambda|\Psi''_\alpha}+O_{\alpha\rightarrow \infty}\left(\alpha^{-3}\right)\leq E_{\alpha,\Lambda}\left(\alpha^2 p\right)+ C\alpha^{-2(1+\sigma)}
\end{align}
for a suitable constant $C>0$, as long as $u+v\leq \frac{1}{2}$. In the following, let $S$ be the set of all $z\in A^3$ such that $\braket{\Psi_z|\mathbb{H}_\Lambda|\Psi_z}>E_{\alpha,\Lambda}\left(\alpha^2 p\right)+\alpha^{-(2+w)}$ for a given $w>0$, and define $M:=\sum_{z\in S}Z_z^2$. By Eq.~(\ref{Equation-IMSforEnergy}), we have
\begin{align*}
M \left(E_{\alpha,\Lambda}\left(\alpha^2 p\right)+\alpha^{-(2+w)}\right)+(1-M)E_\alpha\leq E_{\alpha,\Lambda}\left(\alpha^2 p\right)+C\alpha^{-2(1+\sigma)},
\end{align*}
and therefore $1-M\geq \frac{\alpha^{-(2+w)}-C\alpha^{-2(1+\sigma)}}{E_{\alpha,\Lambda}\left(\alpha^2 p\right)-E_\alpha+\alpha^{-(2+w)}}\geq C_1 \alpha^{-w}$ for $w<2\sigma$, $\alpha$ large enough and a suitable constant $C_1$, where we have used the assumption $E_{\alpha,\Lambda}\left(\alpha^2 p\right)-E_\alpha\lesssim |p|^2\lesssim \alpha^{-2}$. Moreover, let us define $S'$ as the set containing all $z\in A^3$, such that $\braket{\Psi_z|\left(\Upsilon_\Lambda-p\right)^2|\Psi_z}> \alpha^{\frac{1}{2}\left(\frac{9}{5}\sigma+2v+3u-\frac{1}{5}\right)}\alpha^{-2}$ and $M':=\sum_{z\in S'}Z_z^2$. By Eq.~(\ref{Equation-IMSforMomentum}) we see that $M'\leq C_2\alpha^{\frac{1}{2}\left(\frac{9}{5}\sigma+2v+3u-\frac{1}{5}\right)}$ for a suitable constant $C_2$. Consequently
\begin{align*}
\sum_{z\notin S\cup S'}Z_z^2\geq 1-M-M'\geq C_1\alpha^{-w}-C_2\alpha^{\frac{1}{2}\left(\frac{9}{5}\sigma+2v+3u-\frac{1}{5}\right)}
\end{align*}
for $ \alpha$ large enough. Since $\sigma<\frac{1}{9}$, we can take $u,v$ and $w$ small enough, such that $2w+\frac{9}{5}\sigma+2v+3u<\frac{1}{5}$, and consequently $\sum_{z\notin S\cup S'}Z_z^2>0$ for $\alpha$ large enough, which implies the existence of a $z^*\notin S\cup S'$ with $Z_{z_*}>0$, i.e. $\braket{\Psi_{z^*}|\mathbb{H}_\Lambda|\Psi_{z^*}}\leq E_{\alpha,\Lambda}\left(\alpha^2 p\right)+\alpha^{-(2+w)}$ and $\braket{\Psi_{z^*}|\left(\Upsilon_\Lambda-p\right)^2|\Psi_{z^*}}\leq  \alpha^{\frac{1}{2}\left(\frac{9}{5}\sigma+2v+3u-\frac{1}{5}\right)-2}$.

In order to rule out that one of the components $z^*_i$ is infinite, let us verify that $\braket{\Psi_z|\mathbb{H}_\Lambda|\Psi_z}>E_{\alpha,\Lambda}\left(\alpha^2 p\right)+\alpha^{-(2+w)}$ for $\alpha$ large enough in case there exists an $i\in \{1,2,3\}$ with $z_i=\pm \infty$. Note that $\rho\in \mathrm{supp}\left(F_{-\infty,i}\right)$ implies $m_{\alpha^{-v}}(\rho_i)< -\frac{3}{2} L$ and therefore $\int_{|x|> \frac{3}{2}L}\mathrm{d}\rho\geq \int_{-\infty}^{-\frac{3}{2}L}\mathrm{d}\rho_i\geq \int_{-\infty}^{m_{\alpha^{-v}}(\rho_i)}\mathrm{d}\rho_i\geq \left(\frac{1}{2}-\alpha^{-v}\right)\int\mathrm{d}\rho$. Similarly $\int_{|x|> \frac{3}{2}L}\mathrm{d}\rho\geq  \left(\frac{1}{2}-\alpha^{-v}\right)\int\mathrm{d}\rho$ for $\rho\in \mathrm{supp}\left(F_{\infty,i}\right)$. Consequently we have for any $z$ with $z_i=\pm \infty$ for some $i\in \{1,2,3\}$
\begin{align*}
\braket{\Psi_z|\mathcal{N}_{\mathbb{R}^3\setminus B_{\frac{3}{2}L}(0)}|\Psi_z}\geq \left(\frac{1}{2}-\alpha^{-v}\right)\braket{\Psi_z|\mathcal{N}|\Psi_z},
\end{align*}
where $\mathcal{N}_{\mathbb{R}^3\setminus B_{\frac{3}{2}L}(0)}:=\widehat{G}$ with $G(\rho):=\int_{|x|>\frac{3}{2}L}\mathrm{d}\rho$. Therefore \cite[Corollary B.7]{BS1} together with the fact that $\mathrm{supp}\left(\Psi_z\right)\subset \mathrm{supp}\left(\Psi''_\alpha\right)\subset B_L(0)$, yields
\begin{align*}
\braket{\Psi_z|\mathbb{H}_\Lambda|\Psi_z}& \! \geq \! E_\alpha\! +\! \! \left(\! \frac{1}{2}\! -\! \alpha^{-v}\! \! \right)\! \! \braket{\Psi_z|\mathcal{N}|\Psi_z}\! -\! \sqrt{\frac{D}{\frac{3}{2}L\! -\! L}}\! \geq \! E_\alpha\! +\! \left(\frac{1}{2}-\alpha^{-v}\right)c_- \! -\! \sqrt{\! 2D\alpha^{-(1+\sigma)}}\\
&=E_{\alpha,\Lambda}\left(\alpha^2 p\right)+\frac{1}{2}+O_{\alpha \rightarrow \infty}\left(\alpha^{-v}\right)>E_{\alpha,\Lambda}\left(\alpha^2 p\right)+\alpha^{-(2+w)}
\end{align*}
for a suitable constant $D>0$ and $\alpha$ large enough. Hence we obtain that all components $z^*_i$ are finite, i.e. $m_{\alpha^{-v}}(\rho) \in B_{\sqrt{3}\alpha^{-u}}\left(\alpha^{-u}z^*\right)\subseteq \mathbb{R}^3$ for $\rho\in \mathrm{supp}\left(F_{z^*_3,3}F_{z^*_2,2}F_{z^*_1,1}\right)$. 

Let  $\Psi_\alpha''':=\mathcal{T}_{-\alpha^{-u}z^*}\Psi_{z^*}$, where $\mathcal{T}_z$ is a joint translation in the electron and phonon component, i.e. $\left(\mathcal{T}_z \Psi\right)(x):=U_z\Psi(x-z)$ with $U_z$ being defined by $U_z^{-1} a(f)U_z=a(f_z)$ and $f_z(y):=f(y-z)$. Using the fact that $\braket{\Psi_{z^*}|\mathbb{H}_\Lambda|\Psi_{z^*}}\leq E_{\alpha,\Lambda}\left(\alpha^2 p\right)+\alpha^{-(2+w)}\lesssim E_\alpha+\alpha^{-\frac{2}{29}}$ as well as $\mathds{1}_{\Omega^*}\Psi_\alpha'''=\Psi_\alpha'''$, where $\Omega^*$ is the set of all $\rho$ satisfying $\int \mathrm{d}\rho\leq c_+$ and $m_{\alpha^{-v}}\! \left(\rho\right)\in B_{\sqrt{3}\alpha^{-u}}(0)$, we can apply \cite[Lemma 3.11]{BS1}, which yields
\begin{align*}
\Big\langle \Psi'''_\alpha\Big|W_{\varphi^\mathrm{Pek}}^{-1}\, \mathcal{N}W_{\varphi^\mathrm{Pek}}\Big| \Psi'''_\alpha\Big\rangle\lesssim \alpha^{-\frac{2}{29}}+\alpha^{-u}+\alpha^{-v}.
\end{align*}
By taking $r'>0$ small enough such that $r'\leq \frac{1}{2}\left(\frac{1}{5}-\frac{9}{5}\sigma-2v-3u\right)$, $r'\leq w$ and $r'\leq \min\{\frac{2}{29},u,v\}$, we conclude that $\Big\langle \Psi'''_\alpha\Big|W_{\varphi^\mathrm{Pek}}^{-1}\, \mathcal{N}W_{\varphi^\mathrm{Pek}}\Big| \Psi'''_\alpha\Big\rangle\lesssim \alpha^{-r'}$. Since $\mathrm{supp}\left(\Psi'''_\alpha\right)\subset B_L(-\alpha^{-u}z^*)\subset B_{L+\alpha^{-u}|z^*|}(0)\subset B_{4L}(0)$, this concludes the proof.
\end{proof}

In the following Theorem \ref{Theorem: Complete Condensation}, which is the main result of this section, we will lift the (weak) condensation from Eq.~(\ref{Equation-Weak Codensation}) to a strong one without introducing a large energy penalty, using an argument  in \cite{LNSS}. We will verify that the momentum error due to the localization is negligibly small as well. 

\begin{theorem}
\label{Theorem: Complete Condensation}
Given $0<\sigma<\frac{1}{9}$ and $C>0$, let $\Lambda=\alpha^{\frac{4}{5}(1+\sigma)}$ and $L=\alpha^{1+\sigma}$, and assume $p$ satisfies $|p|\leq \frac{C}{\alpha}$ and $E_{\alpha,\Lambda}\! \left(\alpha^2 p\right)\leq E_\alpha+C|p|^2$ for a given $C>0$. Then there exists a $r>0$ and states $\Psi_\alpha$ with $\big\langle \Psi_\alpha\big|\left(\Upsilon_\Lambda-p\right)^2\big|\Psi_\alpha\big\rangle\lesssim \alpha^{-(2+r)}$, $\braket{\Psi_\alpha|\mathbb{H}_\Lambda|\Psi_\alpha}- E_{\alpha,\Lambda}\! \left(\alpha^2 p\right)\lesssim \alpha^{-(2+r)}$ and $\mathrm{supp}\left(\Psi_\alpha\right)\subseteq B_{4L}(0)$, such that
\begin{align}
\label{Equation-StrongCondensation}
\chi\left( W_{\varphi^\mathrm{Pek}-i\xi}^{-1}\, \mathcal{N}W_{\varphi^\mathrm{Pek}-i\xi}\leq \alpha^{-r}\right)\Psi_\alpha=\Psi_\alpha,
\end{align}
where $\xi:=\frac{p}{m}\widetilde{\nabla}_{x_1}\varphi^\mathrm{Pek}$ with $\widetilde{\nabla}_{x_1}:=\chi^1\left(\Lambda^{-1}|\nabla_{x_1}|\leq 2\right)\nabla_{x_1}$.
\end{theorem}
Note that $\xi$ is small in magnitude, $\|\xi\|\lesssim |p|\lesssim \alpha^{-1}$. The statement of Theorem \ref{Theorem: Complete Condensation} is also valid for $\xi=0$, i.e., in case we conjugate by the Weyl transformation $W_{\varphi^\mathrm{Pek}}$ instead of $W_{\varphi^\mathrm{Pek}-i\xi}$. For technical reasons, it will however be useful in the proof of Theorem \ref{Theorem: Parabolic Lower Bound for the truncated Model} to use $\varphi^\mathrm{Pek}-i\xi\approx \varphi^\mathrm{Pek}-i\frac{p}{m}\nabla_{x_1} \varphi^\mathrm{Pek}$ as a reference state, since the latter satisfies the momentum constraint $\braket{\varphi^\mathrm{Pek}-i\frac{p}{m}\nabla_{x_1} \varphi^\mathrm{Pek}|\frac{1}{i}\nabla|\varphi^\mathrm{Pek}-i\frac{p}{m}\nabla_{x_1} \varphi^\mathrm{Pek}}=p$.
\begin{proof}
Let $\Psi_\alpha'''$ be as in Lemma \ref{Lemma-Existence of a condensate} and let us define for $0<\epsilon<\frac{1}{2}$ and $0<h<\min \big\{r',\frac{1}{4}\big\}$
\begin{align*}
\Psi_\alpha:=Z_\alpha^{-1}\chi^\epsilon\left(\alpha^h W_{\varphi^\mathrm{Pek}-i\xi}^{-1}\, \mathcal{N}W_{\varphi^\mathrm{Pek}-i\xi}\leq \frac{1}{2}\right)\Psi_\alpha''',
\end{align*}
where $Z_\alpha:=\|\chi^\epsilon\left(\alpha^h W_{\varphi^\mathrm{Pek}-i\xi}^{-1}\, \mathcal{N}W_{\varphi^\mathrm{Pek}-i\xi}\leq \frac{1}{2}\right)\Psi_\alpha'''\|$ is a normalization constant. Clearly the states $\Psi_\alpha$ satisfy Eq.~(\ref{Equation-StrongCondensation}) for $r\leq h$. Let us furthermore define the states $\widetilde{\Psi}_\alpha:=\frac{1}{\sqrt{1-Z_\alpha^2}}\chi^\epsilon\left(\frac{1}{2}\leq \alpha^h W_{\varphi^\mathrm{Pek}-i\xi}^{-1}\, \mathcal{N}W_{\varphi^\mathrm{Pek}-i\xi}\right)\Psi_\alpha'''$. An application of \cite[Lemma 3.3]{BS1} yields
\begin{align*}
Z_\alpha^2 \braket{\Psi_\alpha|\mathbb{H}_\Lambda|\Psi_\alpha}+&(1-Z_\alpha^2) \braket{\widetilde{\Psi}_\alpha|\mathbb{H}_\Lambda|\widetilde{\Psi}_\alpha}\leq \braket{\Psi'''_\alpha|\mathbb{H}_\Lambda|\Psi'''_\alpha}+C_0\alpha^{2h-\frac{7}{2}}\braket{\Psi'''_\alpha|\sqrt{\mathcal{N}+1}|\Psi'''_\alpha}\\
&\leq E_{\alpha,\Lambda}\! \left(\alpha^2 p\right)+C_1\alpha^{-(2+r'')}
\end{align*}
 for suitable constants $C_0,C_1>0$ and $r'':=\min\{r',\frac{3}{2}-2h\}>0$. We have
 \begin{align*}
&1-Z_\alpha^2=\Big\langle \Psi'''_\alpha\Big|\chi^{\epsilon}\left(\frac{1}{2}\leq  \alpha^{h}W_{\varphi^\mathrm{Pek}-i\xi}^{-1}\mathcal{N}W_{\varphi^\mathrm{Pek}-i\xi} \right)^2\Big| \Psi'''_\alpha\Big\rangle\\
& \leq \frac{2\alpha^{h}}{1-2\epsilon}\Big\langle \Psi'''_\alpha\Big| W_{\varphi^\mathrm{Pek}-i\xi}^{-1}\mathcal{N}W_{\varphi^\mathrm{Pek}-i\xi} \Big|\Psi'''_\alpha \Big\rangle\leq \frac{4\alpha^{h}}{1-2\epsilon}\Big\langle \Psi'''_\alpha\Big| W_{\varphi^\mathrm{Pek}}^{-1}\mathcal{N}W_{\varphi^\mathrm{Pek}} \Big|\Psi'''_\alpha \Big\rangle\! +\! \frac{4\alpha^{h}\|\xi\|^2}{1-2\epsilon}\\
&  \lesssim \frac{1}{1-2\epsilon}\left(\alpha^{h-r'}+\alpha^{h-2}\right)\underset{\alpha\rightarrow \infty}{\longrightarrow}0,
\end{align*}
where we used the operator inequality $W_{\varphi^\mathrm{Pek}-i\xi}^{-1}\mathcal{N}W_{\varphi^\mathrm{Pek}-i\xi}\leq 2\left(W_{\varphi^\mathrm{Pek}}^{-1}\mathcal{N}W_{\varphi^\mathrm{Pek}}+\|\xi\|^2\right)$, $\|\xi\|^2\leq |p|^2\|\nabla \varphi^\mathrm{Pek}\|^2\lesssim \alpha^{-2}$ and Eq.~(\ref{Equation-Weak Codensation}). Making use of $\braket{\widetilde{\Psi}_\alpha|\mathbb{H}_\Lambda|\widetilde{\Psi}_\alpha}\geq E_\alpha$ and $E_{\alpha,\Lambda}\! \left(\alpha^2 p\right)-E_\alpha\lesssim |p|^2\lesssim \alpha^{-2}$, we therefore obtain
\begin{align*}
\braket{\Psi_\alpha|\mathbb{H}_\Lambda|\Psi_\alpha}&-E_{\alpha,\Lambda}\! \left(\alpha^2 p\right)\leq Z_\alpha^{-2}\left(C_1\alpha^{-(2+r'')}+(1-Z_\alpha^2)\left(E_{\alpha,\Lambda}\! \left(\alpha^2 p\right)-E_\alpha\right)\right)\\
& \lesssim \alpha^{-(2+r'')}+\left(\alpha^{h-r'}+\alpha^{h-2}\right)\left(E_{\alpha,\Lambda}\! \left(\alpha^2 p\right)-E_\alpha\right)\lesssim \alpha^{-(2+r''')}
\end{align*}
with $r''':=\min\{r'',r'-h,2-h\}>0$. 

In order to estimate $\big\langle \Psi_\alpha\big|\left(\Upsilon_\Lambda-p\right)^2\big|\Psi_\alpha\big\rangle$, let us apply the IMS identity
\begin{align}
\label{Equation-IMS}
Z_\alpha^2\big\langle \Psi_\alpha\big|\! \left(\Upsilon_\Lambda\! -\! p\right)^2\! \big|\Psi_\alpha\big\rangle\! +\! (1\! -\! Z_\alpha^2)\big\langle \widetilde{\Psi}_\alpha\big|\! \left(\Upsilon_\Lambda\! -\! p\right)^2\! \big|\widetilde{\Psi}_\alpha\big\rangle\! =\! \big\langle \Psi'''_\alpha\big|\! \left(\Upsilon_\Lambda\! -\! p\right)^2\! \big|\Psi'''_\alpha\big\rangle\! -\! \big\langle \Psi'''_\alpha\big|X\big|\Psi'''_\alpha\big\rangle,
\end{align}
where we define $X:=\frac{1}{2}\left[\left[\left(\Upsilon_\Lambda-p\right)^2,A_1\right],A_1\right]+\frac{1}{2}\left[\left[\left(\Upsilon_\Lambda-p\right)^2,A_2\right],A_2\right]$ using the operators $A_1:=f_1\left(W_{\varphi^\mathrm{Pek}-i\xi}^{-1}\, \mathcal{N}W_{\varphi^\mathrm{Pek}-i\xi}\right)$  and $A_2:=f_2\left(W_{\varphi^\mathrm{Pek}-i\xi}^{-1}\, \mathcal{N}W_{\varphi^\mathrm{Pek}-i\xi}\right)$ with $f_1(x):=\chi^\epsilon\big(\alpha^h x\leq \frac{1}{2}\big)$ and $f_2:=\chi^\epsilon\big(\frac{1}{2}\leq \alpha^h x\big)$. In the following let us compute 
\begin{align*}
&\left[\left[\left(\Upsilon_\Lambda\! -\! p\right)^2\! ,A_j\right]\! ,A_j\right]\! \! =\! W_{\varphi^\mathrm{Pek}-i\xi}^{-1}\! \! \left[\left[\! \left(W_{\varphi^\mathrm{Pek}-i\xi}\Upsilon_\Lambda W_{\varphi^\mathrm{Pek}-i\xi}^{-1}-p\right)^2\! \! ,f_j(\mathcal{N})\right]\! ,f_j(\mathcal{N})\right]\!\!  W_{\varphi^\mathrm{Pek}-i\xi}\\
&\ \ \ \ \ =W_{\varphi^\mathrm{Pek}-i\xi}^{-1}\left[\left[\left(\Upsilon_\Lambda-\widetilde{p}+2\mathfrak{Re}\, a^\dagger\left(\varphi\right)\right)^2,f_j(\mathcal{N})\right],f_j(\mathcal{N})\right]W_{\varphi^\mathrm{Pek}-i\xi}
\end{align*}
where  $\varphi:=\frac{1}{i}\widetilde{\nabla}_{x_1}\left(\varphi^\mathrm{Pek}-i\xi\right)$ and $\widetilde{p}:=p-\braket{\varphi^\mathrm{Pek}-i\xi|\frac{1}{i}\widetilde{\nabla}_{x_1}|\varphi^\mathrm{Pek}-i\xi}=p\big(1-\frac{2}{m}\|\widetilde{\nabla}_{x_1}\varphi^\mathrm{Pek}\|^2\big)$. We have $|\widetilde{p}|\leq |p|\leq \frac{C}{\alpha}$ since $m=\frac{2}{3}\|\nabla \varphi^\mathrm{Pek}\|^2=2\|\nabla_{x_1} \varphi^\mathrm{Pek}\|^2\geq 2\|\widetilde{\nabla}_{x_1} \varphi^\mathrm{Pek}\|^2$. Defining the discrete derivative $\delta f_j(x):=\alpha^2\big(f_j(x+\alpha^{-2})-f_j(x)\big)$, we can further write
\begin{align*}
&\left[\left[\left(\Upsilon_\Lambda-\widetilde{p}+2\mathfrak{Re}\, a^\dagger\left(\varphi\right)\right)^2,f_j(\mathcal{N})\right],f_j(\mathcal{N})\right]=8\left[\mathfrak{Re}\, a^\dagger\left(\varphi\right),f(\mathcal{N})\right]^2\\
&\ \ \ \ \ \ +2\Big\{\Upsilon_\Lambda-\widetilde{p}+2\mathfrak{Re}\, a^\dagger\left(\varphi\right),\left[\left[\mathfrak{Re}\, a^\dagger\left(\varphi\right),f_j(\mathcal{N})\right],f_j(\mathcal{N})\right]\Big\}\\
&=-8\alpha^{-4}\left(\mathfrak{Im}\left(a^\dagger\left(\varphi\right)\delta f_j(\mathcal{N})\right)\right)^2+2\alpha^{-4}\Big\{\Upsilon_\Lambda-\widetilde{p}+2\mathfrak{Re}\, a^\dagger\left(\varphi\right),\mathfrak{Re}\left(a^\dagger\left(\varphi\right)(\delta f_j)^2(\mathcal{N})\right)\Big\}
\end{align*}
where we used $\left[\Upsilon_\Lambda-\widetilde{p}+2\mathfrak{Re}\, a^\dagger\left(\varphi\right),f_j(\mathcal{N})\right]=2\left[\mathfrak{Re}\, a^\dagger\left(\varphi\right),f_j(\mathcal{N})\right]$, $\left[\mathfrak{Re}\, a^\dagger\left(\varphi\right),f_j(\mathcal{N})\right]=\alpha^{-2}i\mathfrak{Im}\left(a^\dagger\left(\varphi\right)\delta f_j(\mathcal{N})\right)$ and $\left[\left[\mathfrak{Re}\, a^\dagger\left(\varphi\right),f_j(\mathcal{N})\right],f_j(\mathcal{N})\right]=\alpha^{-4}\mathfrak{Re}\left(a^\dagger\left(\varphi\right)(\delta f_j)^2(\mathcal{N})\right)$. Hence
\begin{align}
\label{Equation-DoubleCommutatorBound}
&- \left[\left[\left(\Upsilon_\Lambda-\widetilde{p}+2\mathfrak{Re}\, a^\dagger\left(\varphi\right)\right)^2,f_j(\mathcal{N})\right],f_j(\mathcal{N})\right]\leq 8\alpha^{-4}\mathfrak{Im}\left(a^\dagger\left(\varphi\right)\delta f_j(\mathcal{N})\right)^2\\
\nonumber
&\ \ \ \ \ \ \ \ + 4\alpha^{-3}\, \mathfrak{Re}\left(a^\dagger\left(\varphi\right)(\delta f_j)^2(\mathcal{N})\right)^2+\alpha^{-5}\left(\Upsilon_\Lambda-\widetilde{p}+2\mathfrak{Re}\, a^\dagger\left(\varphi\right)\right)^2\\
\nonumber
&\leq 2\|\varphi\|^2\left(2\alpha^{-4}\|\delta f_j\|_\infty^2+2\alpha^{-3}\|\delta f_j\|_\infty^4+3\alpha^{-5}\right)\left(2\mathcal{N}+\alpha^{-2}\right)+27\alpha^{-3}\mathcal{N}^2+3\alpha^{-5}|\widetilde{p}|^2
\end{align}
where we have applied multiple Cauchy--Schwarz estimates and used $\Upsilon_\Lambda^2\leq 9\alpha^2 \mathcal{N}^2$. Note that the expression in the last line of Eq.~(\ref{Equation-DoubleCommutatorBound}) is of order $\alpha^{4h-3}\left(\mathcal{N}+1\right)^2$, since $\|\delta f_j\|_\infty\lesssim \alpha^h$ and $\|\varphi\|\lesssim 1$. Using $W_{\varphi^\mathrm{Pek}-i\xi}^{-1}\left(\mathcal{N}+1\right)^2 W_{\varphi^\mathrm{Pek}-i\xi}\lesssim \left(\mathcal{N}+1\right)^2$ we therefore obtain
\begin{align*}
- X=- \frac{1}{2}\sum_{j=1}^2\left[\left[\left(\Upsilon_\Lambda-p\right)^2\! ,A_j\right]\! ,A_j\right]\lesssim \alpha^{4h-3}\left(\mathcal{N}+1\right)^2.
\end{align*}
Using this together with Eq.~(\ref{Equation-IMS}) and the observation $\big\langle \widetilde{\Psi}_\alpha\big|\! \left(\Upsilon_\Lambda\! -\! p\right)^2\! \big|\widetilde{\Psi}_\alpha\big\rangle\geq 0$, yields
\begin{align*}
&\big\langle \Psi_\alpha\big|\! \left(\Upsilon_\Lambda\! -\! p\right)^2\! \big|\Psi_\alpha\big\rangle\leq Z_\alpha^{-2}\left(\big\langle \Psi'''_\alpha\big|\! \left(\Upsilon_\Lambda\! -\! p\right)^2\! \big|\Psi'''_\alpha\big\rangle-\big\langle \Psi'''_\alpha\big|\! X \big|\Psi'''_\alpha\big\rangle\right)\\
&\ \  \lesssim \alpha^{-(2+r')}+\alpha^{4h-3}\big\langle \Psi'''_\alpha\big|\! \left(\mathcal{N}+1\right)^2 \big|\Psi'''_\alpha\big\rangle\lesssim \alpha^{-(2+r')}+\alpha^{4h-3}.
\end{align*}
Since $h<\frac{1}{4}$ we have $\min\{r',1-4h\}>0$, and therefore we can choose $r>0$ small enough such that $r\leq \min\{r',1-4h\}$, $r\leq r'''$ and $r\leq h$, which concludes the proof. 
\end{proof}

\section{Proof of Theorem \ref{Theorem: Parabolic Lower Bound for the truncated Model}}
\label{Section: Proof of Theorem}
In this section we shall prove the main technical Theorem \ref{Theorem: Parabolic Lower Bound for the truncated Model}, using the results of the previous sections as well as the results in the previous part of this paper series \cite{BS1}. Before we do this let us recall some definitions from \cite{BS1}.

\begin{defi}[Finite dimensional Projection $\Pi$]
\label{Definition: Finite dimensional Projection}
Given $\sigma>0$, let $\Lambda:=\alpha^{\frac{4}{5}(1+\sigma)}$ and $\ell:=\alpha^{-4(1+\sigma)}$, and let us introduce the cubes $C_z:=\left[z_1-\ell,z_1+\ell\right)\times \left[z_2-\ell,z_2+\ell\right)\times \left[z_3-\ell,z_3+\ell\right)$ for $z=(z_1,z_2,z_3)\in 2\ell\, \mathbb{Z}^3$. Then we define $\Pi$ as the orthogonal projection onto the subspace spanned by the functions $x\mapsto \int_{C_{z}}\frac{e^{i\, k\cdot x}}{|k|}\,\mathrm{d}k$ for $z\in 2\ell\, \mathbb{Z}^3\setminus \{0\}$ satisfying $C_z\subset B_{\Lambda}(0)$. Furthermore, let $\varphi_1,\dots ,\varphi_N$ be a real orthonormal basis of $\Pi L^2\! \left(\mathbb{R}^3\right)$, such that $\varphi_n=\frac{\Pi \nabla_{x_n}\varphi^\mathrm{Pek}}{\left\|\Pi \nabla_{x_n}\varphi^\mathrm{Pek}\right\|}$ for $n\in \{1,2,3\}$. 
\end{defi}

\begin{defi}[Coordinate Transformation $\tau$]
\label{Definition: Coordinate Transformation}
Let $\varphi^\mathrm{Pek}_x(y):=\varphi^\mathrm{Pek}(y-x)$ and let $t\mapsto x_t$ be the local inverse of the function $x\mapsto (\braket{\varphi_n|\varphi^\mathrm{Pek}_x})_{n=1}^3\in \mathbb{R}^3$ defined for $t\in B_{\delta_*}(0)$ with a suitable $\delta_*>0$. Note that we can take $B_{\delta_*}(0)$ as the domain of the local inverse, since $\braket{\varphi_n|\varphi^\mathrm{Pek}_0}=0$ for all $n\in \{1,2,3\}$ due to the fact that $\varphi^\mathrm{Pek}$ and $\Pi$ respect the reflection symmetry $y_n\mapsto -y_n$. Then we define $f:\mathbb{R}^3\longrightarrow \Pi L^2\! \left(\mathbb{R}^3\right)$ as $f(t):=\chi\left(|t|<\delta_*\right)\left(\Pi\varphi^\mathrm{Pek}_{x_t}-\sum_{n=1}^3 t_n \varphi_n\right)$ and the transformation $\tau:\Pi L^2\! \left(\mathbb{R}^3\right)\longrightarrow \Pi L^2\! \left(\mathbb{R}^3\right)$ as
\begin{align*}
\tau\left(\varphi\right):=\varphi-f\left(t^\varphi\right)
\end{align*}
with $t^\varphi:=\left(\braket{\varphi_1|\varphi},\braket{\varphi_2|\varphi},\braket{\varphi_3|\varphi}\right)\in \mathbb{R}^3$.
\end{defi}

\begin{defi}[Quadratic Approximation $J_{t,\epsilon}$]
\label{Definition: Quadratic Approximation}
Let us first define the operators
\begin{align}
\label{Equation: K}
K^\mathrm{Pek}&:=1-H^\mathrm{Pek}=4\left(-\Delta\right)^{-\frac{1}{2}}\psi^\mathrm{Pek}\frac{1-\ket{\psi^\mathrm{Pek}}\bra{\psi^\mathrm{Pek}}}{-\Delta+V^\mathrm{Pek}-\mu^\mathrm{Pek}}\psi^\mathrm{Pek}\left(-\Delta\right)^{-\frac{1}{2}},\\
\label{Equation: L}
L^\mathrm{Pek}&:=4\left(-\Delta\right)^{-\frac{1}{2}}\psi^\mathrm{Pek}\left(1-\Delta\right)^{-1}\psi^\mathrm{Pek}\left(-\Delta\right)^{-\frac{1}{2}},
\end{align}
where $V^\mathrm{Pek}:=-2(-\Delta)^{-\frac{1}{2}}\varphi^\mathrm{Pek}$, $\mu^\mathrm{Pek}:=e^\mathrm{Pek}-\|\varphi^\mathrm{Pek}\|^2$ and $\psi^\mathrm{Pek}$ is the, non-negative, ground state of the operator $-\Delta+V^\mathrm{Pek}$. Furthermore let $T_x$ be the translation operator, i.e. $\left(T_x\varphi\right)(y):=\varphi(y-x)$, and let $K^\mathrm{Pek}_x:=T_x K^\mathrm{Pek} T_{-x}$ and $L^\mathrm{Pek}_x:=T_x L^\mathrm{Pek} T_{-x}$. Then we define 
\begin{align*}
J_{t,\epsilon}:=\pi \left(1-(1+\epsilon)\left(K^\mathrm{Pek}_{x_t}+\epsilon L^\mathrm{Pek}_{x_t}\right)\right)\pi
\end{align*}
for $|t|<\epsilon$ and $\epsilon<\delta_*$, where $\delta_*$ and $x_t$ are as in Definition \ref{Definition: Coordinate Transformation} and $\pi :L^2\! \left(\mathbb{R}^3\right)\longrightarrow  L^2\! \left(\mathbb{R}^3\right)$ is the orthogonal projection on the space spanned by $\{\varphi_4,\dots ,\varphi_N\}$ with $\varphi_n$ as in Definition \ref{Definition: Finite dimensional Projection}. Furthermore we define $J_{t,\epsilon}:=\pi$ for $|t|\geq \epsilon$ and we will use the shorthand notation $J_{t,\epsilon}[\varphi]:=\braket{\varphi|J_{t,\epsilon}|\varphi}$.
\end{defi}

Recall the definition of $E_{\alpha,\Lambda}$ in Theorem \ref{Theorem: Parabolic Lower Bound for the truncated Model}. 
 In the following we will assume that $p$ satisfies the assumption $E_{\alpha,\Lambda}\! \left(\alpha^2 p\right)\leq E_\alpha+C|p|^2$ of Theorem \ref{Theorem: Complete Condensation} with $C\geq \frac{1}{2m}$, which we can do w.l.o.g., since $E_{\alpha,\Lambda}\! \left(\alpha^2 p\right)> E_\alpha+C|p|^2$ immediately implies the statement of Theorem \ref{Theorem: Parabolic Lower Bound for the truncated Model} (compare with the comment above Lemma \ref{Lemma-InitialState}). We shall also assume in the following that  $|p|\leq \frac{C}{\alpha}$. Due to these assumptions we can apply Theorem \ref{Theorem: Complete Condensation}, which yields the existence of a sequence $\Psi_\alpha$ with $\big\langle \Psi_\alpha\big|\left(\Upsilon_\Lambda-p\right)^2\big|\Psi_\alpha\big\rangle\lesssim \alpha^{-(2+r)}$, $\braket{\Psi_\alpha|\mathbb{H}_\Lambda|\Psi_\alpha}- E_{\alpha,\Lambda}\! \left(\alpha^2 p\right)\lesssim \alpha^{-(2+r)}$ and $\mathrm{supp}\left(\Psi_\alpha\right)\subseteq B_{4L}(0)$ with $L=\alpha^{1+\sigma}$, such that $\widetilde{\Psi}_\alpha:=W_{-i\xi}\Psi_\alpha$ with $\xi=\frac{p}{m}\widetilde{\nabla}_{x_1}\varphi^\mathrm{Pek}$ satisfies condensation with respect to $\varphi^\mathrm{Pek}$, i.e.
\begin{align}
\label{Equation-TildeCondensation}
\chi\left( W_{\varphi^\mathrm{Pek}}^{-1}\mathcal{N} W_{\varphi^\mathrm{Pek}}\leq \alpha^{-r}\right)\widetilde{\Psi}_\alpha=\widetilde{\Psi}_\alpha.
\end{align}
Using $\frac{p}{m}\left(p-\Upsilon_\Lambda \right)\leq \alpha^{-\frac{r}{2}}\frac{|p|^2}{4m^2}+\alpha^{\frac{r}{2}}\left(p-\Upsilon_\Lambda \right)^2$ and $|p|\leq \frac{C}{\alpha}$, we therefore have
\begin{align}
\label{Equation-LagrangeMultiplier}
E_{\alpha,\Lambda}\! \left(\alpha^2 p\right)\geq \Big\langle \Psi_\alpha\Big|\mathbb{H}_\Lambda+\frac{p}{m}\left(p-\Upsilon_\Lambda \right)\Big|\Psi_\alpha\Big\rangle+O_{\alpha\rightarrow \infty}\left(\alpha^{-\left(2+\frac{r}{2}\right)}\right),
\end{align}
where $\frac{p}{m}$ formally acts as a Lagrange multiplier for the minimization of $\mathbb{H}_\Lambda$ subject to the constraint $\Upsilon_\Lambda=p$. 
In the rest of this Section we will verify  that $\mathbb{H}_\Lambda+\frac{p}{m}\left(p-\Upsilon_\Lambda \right)$ is bounded from below by the right hand side of Eq.~(\ref{Equation-Main with cut-off}) when tested against a state $\Psi$ satisfying $\mathrm{supp}\left(\Psi\right)\subseteq B_{4L}(0)$ and complete condensation with respect to $\varphi^\mathrm{Pek}-i\xi$ (where we find it convenient to use $\varphi^\mathrm{Pek}-i\xi$ instead of $\varphi^\mathrm{Pek}$ for technical reasons). The momentum constraint on $\Psi$ will not be needed for this; i.e., 
%
we have transformed our original constrained minimization problem into a global one, which we handle similarly as in the previous part \cite{BS1} concerning a lower bound on the global minimum $E_\alpha=\inf \sigma\left(\mathbb{H}\right)$. As already stressed in the Section \ref{Section: Introduction},  it is essential to work with the truncated Hamiltonian $\mathbb{H}_\Lambda$ and the truncated momentum $\Upsilon_\Lambda$ here, since in contrast to $\mathbb{H}_\Lambda+\frac{p}{m}\left(p-\Upsilon_\Lambda \right)$ the operator $\mathbb{H}+\frac{p}{m}\left(p-\mathbb{P} \right)$ is not bounded from below for $p\neq 0$. 

Following \cite{BS1}, we will identify $\mathcal{F}\left(\Pi L^2\! \left(\mathbb{R}^3\right)\right)$ with $L^2\! \left(\mathbb{R}^{N}\right)$ using the representation of real-valued functions $\varphi=\sum_{n=1}^{N} \lambda_n \varphi_n $ by points $\lambda=(\lambda_1,\dots ,\lambda_{{N}})\in \mathbb{R}^{N}$. With this identification, we can represent the annihilation operators $a_n:=a\left(\varphi_n\right)$ as $a_n=\lambda_n+\frac{1}{2\alpha^2}\partial_{\lambda_n}$, where $\lambda_n$ is the multiplication operator by the function $\lambda\mapsto \lambda_n$ on $L^2\! \left(\mathbb{R}^{N}\right)$. Let us also use for functions $\varphi\mapsto g(\varphi)$ depending on elements $\varphi\in \Pi  L^2\! \left(\mathbb{R}^3\right)$ the convenient notation $g(\lambda):=g\left(\sum_{n=1}^{N} \lambda_n \varphi_n\right)$, where $\lambda\in \mathbb{R}^N$.

It is essential for our proof that $\widetilde{\Psi}_\alpha$ satisfies complete condensation in $\varphi^\mathrm{Pek}$, see Eq.~(\ref{Equation-TildeCondensation}), since it allows us to apply \cite[Lemma 6.1]{BS1} which states that 
in terms of  the quadratic operator $J_{t,\epsilon}$ and the transformation $\tau$ on $\Pi L^2\! \left(\mathbb{R}^3\right)$ in Definitions \ref{Definition: Quadratic Approximation} and \ref{Definition: Coordinate Transformation} we have  
\begin{align}
\label{Equation-FirstPart}
\braket{\widetilde{\Psi}_\alpha|\mathbb{H}_\Lambda|\widetilde{\Psi}_\alpha}\! \geq  & e^\mathrm{Pek}\! +\! \big\langle \widetilde{\Psi}_\alpha\big|\!  -\! \frac{1}{4\alpha^4}\sum_{n=1}^{N}\partial_{\lambda_n}^2\! +J_{t^\lambda,\alpha^{-s}}\! \big[ \tau\! \left(\lambda\right)\big] +\mathcal{N}_{>N}\big| \widetilde{\Psi}_\alpha\big\rangle \! -\! \frac{{N}}{2\alpha^2}\\
\nonumber
&\ \ \ \ \ \ \ \ \ \ \ +  O_{\alpha\rightarrow \infty}\! \left(\! \alpha^{-(2+w)}\! \right)
\end{align}
for suitable $w,s_0>0$ and any $0<s<s_0$, where we define $\mathcal{N}_{>N}:=\mathcal{N}-\sum_{k=1}^{N}a_k^\dagger a_k$ and $t^\varphi$ is defined as in Definition \ref{Definition: Coordinate Transformation} such that  $t^\lambda=(\lambda_1,\lambda_2,\lambda_3)\in \mathbb{R}^3$. Furthermore it is shown in \cite[Lemma 6.1]{BS1}, that there exists a $\beta>0$, such that
\begin{align}
\label{Equation-ExponentialDecay}
\braket{\widetilde{\Psi}_\alpha|1-\mathbb{B}|\widetilde{\Psi}_\alpha}\leq e^{-\beta \alpha^{2-2s}}
\end{align}
for all $0<s<s_0$, where $\mathbb{B}$ is the multiplication operator by the function $\lambda\mapsto \chi(|t^\lambda|<\alpha^{-s})$. In the following we will always choose $s<1$. We will use the symbol $w$ for a generic, positive constant, which is allowed to vary from line to line. 

\subsection{Quasi-Quadratic Lower Bound}
In order to find a good lower bound on $\braket{\Psi_\alpha|\mathbb{H}_\Lambda+\frac{p}{m}\left(p-\Upsilon_\Lambda \right)|\Psi_\alpha}$, and therefore on $E_{\alpha,\Lambda}(\alpha^2p)$, it is natural to conjugate $\mathbb{H}_\Lambda\! +\! \frac{p}{m}\! \left(p\! -\! \Upsilon_\Lambda \right)$ with the Weyl transformation $W_{\varphi^\mathrm{Pek}-i\xi}=W_{\varphi^\mathrm{Pek}}W_{-i\xi}$, since $\varphi^\mathrm{Pek}-i\xi$ is close to the minimizer $\varphi^\mathrm{Pek}-i\frac{p}{m}\nabla_{x_1}\varphi^\mathrm{Pek}$ of the corresponding classical problem, see \cite{FRS}. Since $i\xi$ is purely imaginary, the interaction term in $\mathbb{H}_\Lambda$ is invariant under the transformation $W_{-i\xi}$, i.e. $W_{-i\xi}\mathfrak{Re}\left[a\left(\chi\left(\left|\nabla\right|\leq \Lambda\right)w_x\right)\right]W_{-i\xi}^{-1}=\mathfrak{Re}\left[a\left(\chi\left(\left|\nabla\right|\leq \Lambda\right)w_x\right)\right]$, and furthermore 
\begin{align}
\label{Equation-FirstConjugationOfUpsilon}
W_{-i\xi}\Upsilon_\Lambda W_{-i\xi}^{-1} \! = \! \Upsilon_\Lambda \! - \! 2\mathfrak{Re} \! \left[  a\left(\frac{1}{i}\widetilde{\nabla}_{x_1}i\xi\right) \! \right] \! + \! \Big\langle i\xi\Big|\frac{1}{i}\widetilde{\nabla}_{x_1}\Big|i\xi\Big\rangle \! = \! \Upsilon_\Lambda \! - \! 2\mathfrak{Re} \! \left[  a\left(\widetilde{\nabla}_{x_1}\xi\right) \! \right],
\end{align}
where we have used $\Big\langle i\xi\Big|\frac{1}{i}\widetilde{\nabla}_{x_1}\Big|i\xi\Big\rangle=0$ (since $\braket{h|\frac{1}{i}\widetilde{\nabla}_{x_1}|h}=0$ for any real-valued or imaginary-valued function $h\in L^2\! \left(\mathbb{R}^3\right)$). Therefore conjugating $\mathbb{H}_\Lambda \! +\! \frac{p}{m}\! \left(p\! -\! \Upsilon_\Lambda \right)$ with $W_{-i\xi}$ yields
\begin{align*}
&\Big\langle \Psi_\alpha\Big|\mathbb{H}_\Lambda \! +\! \frac{p}{m}\! \left(p\! -\! \Upsilon_\Lambda \right)\! \Big|\Psi_\alpha\Big\rangle\!  =\! \Big\langle \widetilde{\Psi}_\alpha\Big|\mathbb{H}_\Lambda\! -\! \frac{p}{m}\Upsilon_\Lambda  +\! 2\mathfrak{Re}\left[a\left(\frac{p}{m}\widetilde{\nabla}_{x_1}\xi\! -\! i\xi\right)\right]\! \Big|\widetilde{\Psi}_\alpha\Big\rangle\! +\! \frac{|p|^2}{m} \! +\! \left\|\xi\right\|^2\\
&\ \ \ \ \ \geq e^\mathrm{Pek}\! +\! \big\langle \widetilde{\Psi}_\alpha\big| -\frac{1}{4\alpha^4}\sum_{n=1}^{N}\partial_{\lambda_n}^2\! +\! J_{t^\lambda,\alpha^{-s}}\! \big[ \tau\! \left(\lambda\right)\big] \!+\! \mathcal{N}_{>N}-\frac{p}{m}\Upsilon_\Lambda \big| \widetilde{\Psi}_\alpha\big\rangle-\frac{{N}}{2\alpha^2}\\
&\ \ \ \ \ \ \ \ \ \  \ \ +2\mathfrak{Re}\Big\langle \widetilde{\Psi}_\alpha\Big|a\left(\frac{p}{m}\widetilde{\nabla}_{x_1}\xi\! -\! i\xi\right)\Big|\widetilde{\Psi}_\alpha\Big\rangle+\! \frac{|p|^2}{m} \! +\! \left\|\xi\right\|^2+O_{\alpha\rightarrow \infty}\left(\alpha^{-(2+w)}\right),
\end{align*}
where we have used Eq.~(\ref{Equation-FirstPart}). In the next step we apply the Weyl transformation $W_{\varphi^\mathrm{Pek}}$, which satisfies $W_{\varphi^\mathrm{Pek}}\lambda W_{\varphi^\mathrm{Pek}}^{-1}=\lambda+\lambda^\mathrm{Pek}$ and hence
\begin{align*}
&W_{\varphi^\mathrm{Pek}}\frac{p}{m}\Upsilon_\Lambda W_{\varphi^\mathrm{Pek}}^{-1}=\frac{p}{m}\Upsilon_\Lambda+2\mathfrak{Re}\left[a\left(\frac{p}{im}\widetilde{\nabla}_{x_1}\varphi^\mathrm{Pek}\right)\right]=\frac{p}{m}\Upsilon_\Lambda-2\mathfrak{Re}\left[a\left(i\xi\right)\right],\\
&W_{\varphi^\mathrm{Pek}}\mathfrak{Re}\left[a\left(\frac{p}{m}\widetilde{\nabla}_{x_1}\xi\! -\! i\xi\right)\right] W_{\varphi^\mathrm{Pek}}^{-1}=\mathfrak{Re}\left[a\left(\frac{p}{m}\widetilde{\nabla}_{x_1}\xi\! -\! i\xi\right)\right]-\|\xi  \|^2,
\end{align*}
where we have used $\mathfrak{Re}\braket{\varphi^\mathrm{Pek}|\frac{p}{m}\widetilde{\nabla}_{x_1}\xi\! -\! i\xi}=\braket{\varphi^\mathrm{Pek}|\frac{p}{m}\widetilde{\nabla}_{x_1}\xi}=-\|\xi\|^2$. Furthermore $W_{\varphi^\mathrm{Pek}}t^\lambda W_{\varphi^\mathrm{Pek}}^{-1}=(\lambda_1+\lambda^\mathrm{Pek}_1,\lambda_2+\lambda^\mathrm{Pek}_2 ,\lambda_3+\lambda^\mathrm{Pek}_3)=(\lambda_1,\lambda_2 ,\lambda_3)=t^\lambda$ with $\lambda^\mathrm{Pek}:=\left(\braket{\varphi_n|\Pi \varphi^\mathrm{Pek}}\right)_{n=1}^N$. Therefore defining $\Psi_\alpha^*:=W_{\varphi^\mathrm{Pek}}\widetilde{\Psi}_\alpha=W_{\varphi^\mathrm{Pek}-i\xi}\Psi_\alpha$ and conjugating with $W_{\varphi^\mathrm{Pek}}$ yields the  lower bound 
\begin{align}
\label{Equation-LowerBoundWithMultiplierAlt}
&\braket{\Psi_\alpha|\mathbb{H}_\Lambda+\frac{p}{m}\left(p-\Upsilon_\Lambda \right)|\Psi_\alpha} \\ \nonumber
&\geq e^\mathrm{Pek}\! +\!   \Big\langle \Psi^*_\alpha\Big| -\! \frac{1}{4\alpha^4}\sum_{n=1}^{N}\partial_{\lambda_n}^2\!  +\!  J_{t^\lambda,\alpha^{-s}}\! \left[ \tau\! \left(\lambda\! +\! \lambda^\mathrm{Pek}\right)\right]\!  +\! W_{\varphi^\mathrm{Pek}}\mathcal{N}_{>N}W_{\varphi^\mathrm{Pek}}^{-1}\! -\! \frac{p}{m}\Upsilon_\Lambda\Big| \Psi^*_\alpha\Big\rangle\!\\
\nonumber
&\quad  - \frac{{N}}{2\alpha^2} +2\mathfrak{Re}\Big\langle \Psi^*_\alpha\Big|a\left(\frac{p}{m}\widetilde{\nabla}_{x_1}\xi\right)\Big| \Psi^*_\alpha \Big\rangle\!+\! \frac{|p|^2}{m}\! -\! \|\xi\|^2 \!+ \! O_{\alpha\rightarrow \infty}\! \left(\alpha^{-(2+w)}\right).
\end{align}
The advantage of conjugating with the Weyl transformation $W_{\varphi^\mathrm{Pek}-i\xi}=W_{\varphi^\mathrm{Pek}}W_{-i\xi}$ stems from the observation that we have an almost complete cancellation of linear terms, i.e., as we will verify below, the term linear in creation and annihilation operators $\mathfrak{Re}\Big\langle \Psi^*_\alpha\Big|a\left(\frac{p}{m}\widetilde{\nabla}_{x_1}\xi\right)\Big| \Psi^*_\alpha \Big\rangle$ in Eq.~(\ref{Equation-LowerBoundWithMultiplierAlt}) is of negligible order, and  the function $\lambda\mapsto J_{t^\lambda,\alpha^{-s}}\! \left[ \tau\! \left(\lambda\! +\! \lambda^\mathrm{Pek}\right)\right]$ vanishes  quadratically at $\lambda=0$. The latter follows from the fact that $\tau\! \left( \lambda^\mathrm{Pek}\right)=0$. Utilizing the inequalities $\braket{\Psi^*_\alpha|\mathcal{N}|\Psi^*_\alpha}= \braket{\Psi_\alpha|W_{\varphi^\mathrm{Pek}-i\xi}^{-1}\mathcal{N}W_{\varphi^\mathrm{Pek}-i\xi}|\Psi_\alpha}\leq \alpha^{-r}$, see Eq.~(\ref{Equation-StrongCondensation}), and $\|\frac{p}{ m}\widetilde{\nabla}_{x_1}\xi\|\lesssim |p|^2$, where we have used that $\varphi^\mathrm{Pek}\in H^2\! \left(\mathbb{R}^3\right)$, see \cite{Li,MS}, we obtain that
\begin{align}
\label{Equation-LinerTerm}
2\mathfrak{Re}\Big\langle \Psi^*_\alpha \Big|a\left(\frac{p}{m}\widetilde{\nabla}_{x_1}\xi\right)\Big|\Psi^*_\alpha\Big\rangle \lesssim\!  \alpha^{-\frac{r}{2}}|p|^2\! \lesssim \! \alpha^{-(2+\frac{r}{2})}
\end{align}
is indeed negligible small. Furthermore we can estimate, up to a term of order $\alpha^{-(2+\frac{2}{5})}$, $W_{\varphi^\mathrm{Pek}}\mathcal{N}_{>N}W_{\varphi^\mathrm{Pek}}^{-1}$ from below by a proper quadratic expression
\begin{align}
\nonumber
&W_{\varphi^\mathrm{Pek}}\mathcal{N}_{>N}W_{\varphi^\mathrm{Pek}}^{-1}=\mathcal{N}_{>N}+a\left((1-\Pi)\varphi^\mathrm{Pek}\right)+a^\dagger\left((1-\Pi)\varphi^\mathrm{Pek}\right)+\left\|(1-\Pi)\varphi^\mathrm{Pek}\right\|^2\\
\label{Equation-TransformedNumberOperator}
& \ \ \ \ \ \ \geq \frac{1}{2}\mathcal{N}_{>N}-2\left\|(1-\Pi)\varphi^\mathrm{Pek}\right\|^2=\frac{1}{2}\mathcal{N}_{>N}+O_{\alpha\rightarrow \infty}\left(\alpha^{-(2+\frac{2}{5})}\right),
\end{align}
where we have used $\|(1-\Pi)\varphi^\mathrm{Pek}\|^2\lesssim \alpha^{-(2+\frac{2}{5})}$, see \cite[Lemma A.1]{BS1}. In the following let us use the convenient notation $e^\mathrm{Pek}_p:=e^\mathrm{Pek}+\frac{|p|^2}{2m}$. Combining Eq.~(\ref{Equation-LowerBoundWithMultiplierAlt}) with Eq.~(\ref{Equation-LinerTerm}), Eq.~(\ref{Equation-TransformedNumberOperator}) and the observation that $\frac{|p|^2}{m}\! -\! \|\xi\|^2\geq \frac{|p|^2}{2m}$, and using the fact that $E_{\alpha,\Lambda}(\alpha^2p)\geq \Big\langle \Psi_\alpha\Big|\mathbb{H}_\Lambda \! +\! \frac{p}{m}\! \left(p\! -\! \Upsilon_\Lambda \right)\! \Big|\Psi_\alpha\Big\rangle+O_{\alpha\rightarrow \infty}\left(\alpha^{-\left(2+\frac{r}{2}\right)}\right)$, see Eq.~(\ref{Equation-LagrangeMultiplier}), we obtain
\begin{align}
\nonumber
E_{\alpha,\Lambda}(\alpha^2p)&\geq\,   e^\mathrm{Pek}_p\! + \!  \Big\langle \Psi^*_\alpha\Big|\!-\frac{1}{4\alpha^4}\sum_{n=1}^{N}\partial_{\lambda_n}^2\!  +\!  J_{t^\lambda,\alpha^{-s}}\! \left[ \tau\! \left(\lambda\! +\! \lambda^\mathrm{Pek}\right)\right]+\!  \frac{1}{2}\mathcal{N}_{>N}-\! \frac{p}{m}  \Upsilon_\Lambda\Big| \Psi^*_\alpha\Big\rangle\\ 
\label{Equation-PseudoQuadratic}
& \ \ \ \ \ \ \ \ \ \ \  \ \ \ \ \ \ \  -\frac{{N}}{2\alpha^2}+O_{\alpha\rightarrow \infty}\left(\alpha^{-(2+w)}\right).
\end{align}
The right hand side of Eq.~(\ref{Equation-PseudoQuadratic}) is up to a coordinate transformation in the argument of $J_{t^\lambda,\alpha^{-s}}$ quadratic in creation and annihilation operators. In the next subsection we will apply a unitary transformation in order to arrive at a proper quadratic expression.

\subsection{Conjugation with the Unitary $\mathcal{U}$}
In order to get rid of the coordinate transformation $\tau$ in the argument of $J_{t^\lambda,\alpha^{-s}}$, let us define the unitary operator $\mathcal{U}$ on $\mathcal{F}\left(\Pi L^2\! \left(\mathbb{R}^3\right)\right)\cong L^2\! \left(\mathbb{R}^{N}\right)$ as $\mathcal{U}\left(\Psi\right)(\lambda):=\Psi\left(\Xi(\lambda)\right)$, where $\Xi:\mathbb{R}^N\longrightarrow \mathbb{R}^N$ is defined as $\Xi(\lambda):=\tau\Big(\lambda + \lambda^\mathrm{Pek}\Big) \in \Pi L^2\! \left(\mathbb{R}^3\right) \cong \mathbb{R}^N$.  Note that the inverse of $\tau$ is simply given by $\tau^{-1}(\varphi)=\varphi+f(t^\varphi)$ where $f:\mathbb{R}^3\longrightarrow \Pi L^2\! \left(\mathbb{R}^3\right)$ is defined in Definition \ref{Definition: Coordinate Transformation}, which can be checked easily using the fact that $\braket{\varphi_n|f(t)}=0$ for $n\in \{1,2,3\}$ and consequently $t^{\tau(\varphi)}=t^\varphi$. Hence
\begin{align}
\label{Equation: LambdaTransf}
\mathcal{U}^{-1} \lambda_n \,  \mathcal{U}=\braket{\varphi_n|\tau^{-1}(\lambda)}-\lambda^\mathrm{Pek}_n=\lambda_n+\braket{\varphi_n|f(t^\lambda)}-\lambda^\mathrm{Pek}_n
\end{align}
and therefore $\mathcal{U}^{-1} t^\lambda\,  \mathcal{U}=(\braket{\varphi_1|\tau^{-1}(\lambda)}-\lambda^\mathrm{Pek}_1,\dots ,\braket{\varphi_3|\tau^{-1}(\lambda)}-\lambda^\mathrm{Pek}_3)=(\lambda_1,\dots ,\lambda_3)=t^\lambda$. Defining the matrix $\left(J_{t,\epsilon}\right)_{n,m}:=\braket{\varphi_n|J_{t,\epsilon}|\varphi_m}$ we furthermore have
\begin{align*}
&\mathcal{U}^{-1} J_{t^\lambda,\alpha^{-s}}\! \big[ \tau\! \left(\lambda\! +\! \lambda^\mathrm{Pek}\right)\big]\,  \mathcal{U}=J_{t^\lambda,\alpha^{-s}}\! \big[ \lambda\big]=\sum_{n,m=4}^N \left(J_{t^\lambda,\alpha^{-s}}\right)_{n,m}\lambda_n \lambda_m
\end{align*}
as well as $\mathcal{U}^{-1} i\partial_{\lambda_n}\,  \mathcal{U}=i\partial_{\lambda_n}$ for $3<n\leq N$, which immediately follows from the observation that $\Xi$ is a $t^\lambda=(\lambda_1,\lambda_2,\lambda_3)$-dependent shift. In the following let us extend $\{\varphi_1,\dots,\varphi_N\}$ to an orthonormal basis $\{\varphi_n:n\in \mathbb{N}\}$ of $L^2\! \left(\mathbb{R}^3\right)$ and introduce $a_n:=a\left(\varphi_n\right)$ for all $n\in \mathbb{N}$, and let us extend the action of $\mathcal{U}$ to all of $\mathcal{F}\left(L^2\! \left( \mathbb{R}^3\right)\right)$ such that $\mathcal{U}^{-1} a_n\,  \mathcal{U}=a_n$ for $n>N$. Defining $\Psi_\alpha':=\mathcal{U}^{-1}\Psi_\alpha^*$, we obtain by Eq.~(\ref{Equation-PseudoQuadratic})
\begin{align}
\nonumber
E_{\alpha,\Lambda}(\alpha^2p)  \! \geq \!   e^\mathrm{Pek}_p \! &+ \! \Big\langle \Psi'_\alpha\Big| \! -\! \frac{1}{4\alpha^4}\sum_{n=1}^3 \mathcal{U}^{-1}\partial_{\lambda_n}^2\mathcal{U}-\frac{1}{4\alpha^4}\sum_{n=4}^N \partial_{\lambda_n}^2\! +\! \sum_{n,m=4}^N \left(J_{t^\lambda,\alpha^{-s}}\right)_{n,m}\lambda_n \lambda_m\!   \\
\label{Equation-RepresentationOfTransformatedOperator}
& \ \ +\frac{1}{2}\mathcal{N}_{>N}-\mathcal{U}^{-1} \frac{p}{m}  \Upsilon_\Lambda\,  \mathcal{U}\Big| \Psi'_\alpha\Big\rangle
-\frac{{N}}{2\alpha^2}+O_{\alpha\rightarrow \infty}\left(\alpha^{-(2+w)} \right).
\end{align}
Using Eq.~(\ref{Equation: LambdaTransf}) and $\mathcal{U}^{-1} i\partial_{\lambda_n}\,  \mathcal{U}=i\partial_{\lambda_n}$ for $3<n\leq N$, we further obtain the transformation law $\mathcal{U}^{-1} a_n\,  \mathcal{U}=a_n+\braket{\varphi_n|f(t^\lambda)-\Pi\varphi^\mathrm{Pek}}$ for all $n>3$. 

In order to express $\mathcal{U}^{-1} \frac{p}{m}  \Upsilon_\Lambda\,  \mathcal{U}$, let us introduce the operators $c_n$ defined as $c_n:=\frac{1}{2\alpha^2} \mathcal{U}^{-1} \partial_{\lambda_n}\,  \mathcal{U}$ for $n\in \{1,2,3\}$ and $c_n:=a_n$ for $n>3$, as well as $g(t):=f(t)-\Pi\varphi^\mathrm{Pek}+\sum_{n=1}^3 t_n \varphi_n\in \Pi L^2\! \left(\mathbb{R}^3\right)$ and $g_n(t):=\braket{\varphi_n|g(t)}$. With these definitions at hand we obtain
\begin{align*}
\mathcal{U}^{-1} a_n\,  \mathcal{U}&=\mathcal{U}^{-1}\left( \frac{1}{2\alpha^2}\partial_{\lambda_n}+ \lambda_n\right)  \mathcal{U}=\frac{1}{2\alpha^2} \mathcal{U}^{-1} \partial_{\lambda_n}\,  \mathcal{U}+ \lambda_n=c_n+g_n\! \left(t^\lambda\right)\text{, for }1\leq n\leq 3,\\
\mathcal{U}^{-1} a_n\,  \mathcal{U}&=a_n+\braket{\varphi_n|f(t^\lambda)-\Pi\varphi^\mathrm{Pek}}=c_n+g_n\! \left(t^\lambda\right)\text{, for }4\leq n\leq N
\end{align*}
and $\mathcal{U}^{-1} a_n\,  \mathcal{U}=c_n=c_n+g_n\! \left(t^\lambda\right)$ for $n>N$, and therefore $\mathcal{U}^{-1} a_n\,  \mathcal{U}=c_n+g_n\! \left(t^\lambda\right)$ for all $n\in \mathbb{N}$. In the following we want to think of $c_n$ as being a variable of magnitude $\alpha^{-1}$ and $t^\lambda$ as being of order $\alpha^{-r}$ for some $r>0$, and consequently we think of $g_n\! \left(t^\lambda\right)$ as being of order $\alpha^{-r}$ as well, since $g(0)=0$. While the former will be a consequence of the proof presented below, the control on $t^\lambda$ follows from our assumption that we have condensation with respect to the state $\varphi^\mathrm{Pek}$.

In the following we want to show that for suitable $w,w'>0$, $ \frac{p}{m}\Upsilon_\Lambda$ is bounded by $\epsilon\left(- \frac{1}{4\alpha^4}\sum_{n=1}^3 \mathcal{U}^{-1}\partial_{\lambda_n}^2\mathcal{U}+\sum_{n=4}^N a_n^\dagger a_n+\mathcal{N}_{>N}\right)$ with $\epsilon=\alpha^{-w'}$, up to a term of negligible magnitude, see Eq.~(\ref{Equation-FinalEstimateUpsilon}). Since $- \frac{1}{4\alpha^4}\sum_{n=1}^3 \mathcal{U}^{-1}\partial_{\lambda_n}^2\mathcal{U}$ and $\mathcal{N}_{>N}$ appear in the expression on the right hand side of Eq.~(\ref{Equation-RepresentationOfTransformatedOperator}) as well, and since they are non-negative, this will leave us with the study of $-\frac{1}{4\alpha^4}\sum_{n=4}^N \partial_{\lambda_n}^2\! +\! \sum_{n,m=4}^N \left(J_{t^\lambda,\alpha^{-s}}\right)_{n,m}\lambda_n \lambda_m-\epsilon \sum_{n=4}^N a_n^\dagger a_n$ for a lower bound on the expression on the right hand side of Eq.~(\ref{Equation-RepresentationOfTransformatedOperator}). Using the representation $\frac{p}{m}  \Upsilon_\Lambda=\sum_{n,m=1}^\infty \braket{\varphi_n|\frac{p}{i\, m}\widetilde{\nabla}_{x_1}|\varphi_m}a_n^\dagger a_m$, we obtain
\begin{align}
\nonumber
&\mathcal{U}^{-1} \frac{p}{m}  \Upsilon_\Lambda\,  \mathcal{U}\! =\! \! \! \! \sum_{n,m=1}^\infty \! \! \braket{\varphi_n|\frac{p}{i\, m}\widetilde{\nabla}_{x_1}|\varphi_m}\left(c_n+g_n\! \left(t^\lambda\right)\right)^\dagger \left(c_m+g_m\! \left(t^\lambda\right)\right)\\
\label{Equation-Transformed Upsilon}
 &\ \ \ \ =\! \! \! \! \sum_{n,m=1}^\infty\! \!  \braket{\varphi_n|\frac{p}{i\, m}\widetilde{\nabla}_{x_1}|\varphi_m} c_n^\dagger c_m
 +\sum_{n,m=1}^\infty \braket{\varphi_n|\frac{p}{i\, m}\widetilde{\nabla}_{x_1}|\varphi_m}\left(c_n^\dagger\, g_m\! \left(t^\lambda\right)+g_n\! \left(t^\lambda\right)\, c_m \right),
\end{align}
where we have used $\sum_{n,m=1}^\infty \braket{\varphi_n|\frac{p}{i\, m}\widetilde{\nabla}_{x_1}|\varphi_m}g_n\! \left(t^\lambda\right)g_m\! \left(t^\lambda\right)=\braket{g\! \left(t^\lambda\right)|\frac{p}{i\, m}\widetilde{\nabla}_{x_1}|g\! \left(t^\lambda\right)}=0$, see the comment below Eq.~(\ref{Equation-FirstConjugationOfUpsilon}). Using the bound on the operator norm $\|\frac{p}{m}\widetilde{\nabla}_{x_1}\|_\mathrm{op}\leq \frac{|p|}{m}3\Lambda=\frac{|p|}{m}3\alpha^{\frac{4}{5}(1+\sigma)}\lesssim \alpha^{\frac{4}{5}(1+\sigma)-1}$ yields
\begin{align}
\label{Equation: TruncatedMomentumImplication}
\pm \sum_{n,m=1}^\infty \braket{\varphi_n|\frac{p}{i\, m}\widetilde{\nabla}_{x_1}|\varphi_m} c_n^\dagger c_m\lesssim \alpha^{\frac{4}{5}(1+\sigma)-1}\sum_{n=1}^\infty c_n^\dagger c_n.
\end{align}
For the bound in Eq.~(\ref{Equation: TruncatedMomentumImplication}) it is essential that we are using the truncated momentum $\Upsilon_\Lambda$ defined in terms of the bounded operator $\widetilde{\nabla}_{x_1}$ instead of the unbounded operator $\nabla_{x_1}$. Defining the coefficients $h_n(t):=\sum_{m=1}^\infty \braket{\varphi_n|\frac{p}{i\, m}\widetilde{\nabla}_{x_1}|\varphi_m}g_m(t)$ and applying Cauchy--Schwarz furthermore yields for all $\epsilon>0$
\begin{align*}
\pm \sum_{n,m=1}^\infty &\braket{\varphi_n|\frac{p}{i\, m}\widetilde{\nabla}_{x_1}|\varphi_m}\left(c_n^\dagger\, g_m\! \left(t^\lambda\right)+g_n\! \left(t^\lambda\right)\, c_m \right)=\sum_{n=1}^\infty \left(c_n^\dagger h_n\! \left(t^\lambda\right)+\overline{h_n\! \left(t^\lambda\right)}c_n\right)\\
&\leq \epsilon \sum_{n=1}^\infty c_n^\dagger c_n+\epsilon^{-1}\sum_{n=1}^\infty \left|h_n\! \left(t^\lambda\right)\right|^2=\epsilon \sum_{n=1}^\infty c_n^\dagger c_n+\epsilon^{-1}\left\|\frac{p}{m}\widetilde{\nabla}_{x_1} g\! \left(t^\lambda\right)\right\|^2.
\end{align*}
Note that $\left\|\frac{p}{ m}\widetilde{\nabla}_{x_1} g(t)\right\|\leq \frac{|p|}{m}\left\|\nabla g(t)\right\|$. Making use of $\nabla g(t)=\nabla \Pi \eta(t)$ with
\begin{align*}
\eta(t):=\chi\left(|t|<\delta_*\right)\left(\varphi^\mathrm{Pek}_{x_t}-\varphi^\mathrm{Pek}\right)+\chi\left(\delta_*\leq |t|\right)\left(\sum_{n=1}^3 t_n \frac{\nabla_{x_n}\varphi^\mathrm{Pek}}{\|\Pi \nabla_{x_n}\varphi^\mathrm{Pek}\|}-\varphi^\mathrm{Pek}\right),
\end{align*}
we obtain $\|\nabla g(t)\|\lesssim \|\nabla \eta(t)\|+\alpha^{-4(1+\sigma)}\|\eta(t)\|$ by Lemma \ref{Lemma: Compatibility}. Using again $\varphi^\mathrm{Pek}\in H^2\! \left(\mathbb{R}^3\right)$, we have $\|\eta(t)\|+\|\nabla \eta(t)\|\lesssim 1+|t|$, as well as $\|\nabla \eta(t)\|=\|\nabla \varphi^\mathrm{Pek}_{x_t}-\nabla \varphi^\mathrm{Pek}\|\leq |x_t|\|\Delta \varphi^\mathrm{Pek}\|\lesssim |t|$ for $|t|<\delta_*$. Consequently, $\left\|\frac{p}{ m}\widetilde{\nabla}_{x_1} g(t)\right\|\leq C_0 |p| \left(|t|+\alpha^{-4(1+\sigma)}(1+|t|)\right)$ for a suitable constant $C_0$. The choice $\epsilon:=\alpha^{-\min\{\frac{r}{2},1\}}$ yields for $\alpha$ large enough 
\begin{align}
\label{Equation-FinalEstimateUpsilon}
\pm \mathcal{U}^{-1} \frac{p}{m}  \Upsilon_\Lambda\,  \mathcal{U}\leq \alpha^{-w'}\sum_{n=1}^\infty c_n^\dagger c_n+ C_0C^2 \left(\alpha^{-2}\alpha^{\min\{\frac{r}{2},1\}}\left|t^\lambda\right|^2+\alpha^{-5-4\sigma}\Big(1+\left|t^\lambda\right|\Big)^2\right)
\end{align}
with $w'<\min\{\frac{r}{2},1-\frac{4}{5}(1+\sigma)\}$. In the following let $\alpha$ be large enough such that $\alpha^{-w'}\leq \frac{1}{2}$. Then we have 
\begin{align*}
\alpha^{-w'}\! \! \! \! \sum_{n\notin \{4,\dots,N\}}\! \!  c_n^\dagger c_n=\alpha^{-w'}\left(\sum_{n>N}a_n^\dagger a_n- \frac{1}{4\alpha^4}\sum_{n=1}^3 \mathcal{U}^{-1}\partial_{\lambda_n}^2\mathcal{U}\right)\! \leq \frac{1}{2}\mathcal{N}_{>N}\! -\!  \frac{1}{4\alpha^4}\sum_{n=1}^3 \mathcal{U}^{-1}\partial_{\lambda_n}^2\mathcal{U}.
\end{align*}
Using Eq.~(\ref{Equation-RepresentationOfTransformatedOperator}), Eq.~(\ref{Equation-FinalEstimateUpsilon}) and $\braket{\Psi'_\alpha||t^\lambda|^2|\Psi'_\alpha}= \braket{\widetilde{\Psi}_\alpha||t^\lambda|^2|\widetilde{\Psi}_\alpha}\leq \braket{\widetilde{\Psi}_\alpha|\mathcal{N}|\widetilde{\Psi}_\alpha}+\frac{3}{2\alpha^2}\leq \alpha^{-r}+\frac{3}{2\alpha^2}$, see Theorem \ref{Theorem: Complete Condensation} for the last estimate, we obtain for a suitable $w>0$
\begin{align*}
E_{\alpha,\Lambda}(\alpha^2p) &\geq e^\mathrm{Pek}_p+\Big\langle \Psi'_\alpha\Big|-\frac{1}{4\alpha^4}\sum_{n=4}^N \partial_{\lambda_n}^2\! +\! \sum_{n,m=4}^N \left(J_{t^\lambda,\alpha^{-s}}\right)_{n,m}\lambda_n \lambda_m\\
&\ \ \ \ \ \ \ \  \ \ \ \ \  \ \ \ \ \ -\alpha^{-w'}\sum_{n=4}^N a_n^\dagger a_n\Big| \Psi'_\alpha\Big\rangle-\frac{N}{2\alpha^2} +O_{\alpha\rightarrow \infty}\left(\alpha^{-(2+w)}\right)\\
& =   e^\mathrm{Pek}_p+\left(1-\alpha^{-w'}\right)\Big\langle \Psi'_\alpha\Big| \mathbb{Q}^{\alpha^{-w'}}_{t^\lambda,\alpha^{-s}}-\frac{N}{2\alpha^2}\Big| \Psi'_\alpha\Big\rangle +O_{\alpha\rightarrow \infty}\left(\alpha^{-(2+w)}\right)
\end{align*}
with $\mathbb{Q}^\kappa_{t,\epsilon}:= -\frac{1}{4\alpha^4}\sum_{n=4}^N \partial_{\lambda_n}^2+\frac{1}{1-\kappa}\sum_{n,m=4}^N \left(\left(J_{t,\epsilon}\right)_{n,m}-\kappa \delta_{n,m}\right)\lambda_n \lambda_m$, where we made use of the fact that $\sum_{n=4}^N a_n^\dagger a_n=-\frac{1}{4\alpha^4}\sum_{n=4}^N \partial_{\lambda_n}^2+\sum_{n=4}^N \lambda_n^2-\frac{N-3}{2\alpha^2}$.

\subsection{Properties of the Harmonic Oscillators $\mathbb{Q}^\kappa_{t,\epsilon}$}
Let $\pi$ be the projection from Definition \ref{Definition: Quadratic Approximation} and note that $J_{t,\epsilon}\geq c\, \pi$ for suitable $c>0$, $\epsilon$ small enough and $\alpha$ large enough by \cite[Lemma B.5]{BS1}. Therefore $\mathbb{Q}^{\alpha^{-w'}}_{t,\alpha^{-s}}\geq 0$ for $\alpha$ large enough. Since $J_{t,\epsilon}\leq 1$, we furthermore have $(1-\kappa)\inf \sigma \left(\mathbb{Q}^\kappa_{t,\epsilon}\right)\leq \frac{N}{2\alpha^2}\lesssim \alpha^{-2}\left(\frac{\Lambda}{\ell}\right)^3\leq \alpha^q$ for a suitable exponent $q$, see Definition \ref{Definition: Finite dimensional Projection}. Combining this with the estimate $\braket{\Psi'_\alpha|1-\mathbb{B}|\Psi'_\alpha}=\braket{\widetilde{\Psi}_\alpha|1-\mathbb{B}|\widetilde{\Psi}_\alpha}\leq e^{-\beta \alpha^{2-2s}}$ for a suitable $\beta>0$, where $\mathbb{B}:=\chi(|t^\lambda|<\alpha^{-s})$, see Eq.~(\ref{Equation-ExponentialDecay}), yields 
$$
\inf_{|t|<\alpha^{-s}}\inf \sigma\left(\mathbb{Q}^{\alpha^{-w'}}_{t,\alpha^{-s}}\right)\braket{\Psi_\alpha|\mathbb{B}|\Psi_\alpha}\geq \inf_{|t|<\alpha^{-s}}\inf \sigma\left(\mathbb{Q}^{\alpha^{-w'}}_{t,\alpha^{-s}}\right)+O_{\alpha\rightarrow \infty}\left(\alpha^q e^{-\beta \alpha^{2-2s}}\right)\,.
$$
Therefore we obtain for a suitable $w>0$
\begin{align}
\nonumber
&E_{\alpha,\Lambda}(\alpha^2p)  \geq  e^\mathrm{Pek}_p+\left(1-\alpha^{-w'}\right)\Big\langle \Psi'_\alpha\Big|  \mathbb{Q}^{\alpha^{-w'}}_{t^\lambda,\alpha^{-s}}\mathbb{B}-\frac{N}{2\alpha^2}\Big| \Psi'_\alpha\Big\rangle +O_{\alpha\rightarrow \infty}\left(\alpha^{-(2+w)}\right)\\
\nonumber
& \ \ \geq e^\mathrm{Pek}_p+\left(1-\alpha^{-w'}\right)\left(\inf_{|t|<\alpha^{-s}}\inf \sigma\left(\mathbb{Q}^{\alpha^{-w'}}_{t,\alpha^{-s}}\right)\braket{\Psi_\alpha|\mathbb{B}|\Psi_\alpha}-\frac{N}{2\alpha^2}\right) +O_{\alpha\rightarrow \infty}\left(\alpha^{-(2+w)}\right)\\
\label{Equation-MinimalBogoliubovEnergy}
&\ \ \geq e^\mathrm{Pek}_p+\left(1-\alpha^{-w'}\right)\left(\inf_{|t|<\alpha^{-s}}\inf \sigma\left(\mathbb{Q}^{\alpha^{-w'}}_{t,\alpha^{-s}}\right)-\frac{N}{2\alpha^2}\right) +O_{\alpha\rightarrow \infty}\left(\alpha^{-(2+w)}\right).
\end{align}
Since $\mathbb{Q}^\kappa_{t,\epsilon}$ is a harmonic oscillator, we can write its ground state energy explicitly as 
\begin{align*}
\inf \sigma\left( \mathbb{Q}^\kappa_{t,\epsilon}\right)& =\frac{1}{2\alpha^2}\mathrm{Tr}_{\Pi L^2(\mathbb{R}^3)} \sqrt{\frac{J_{t,\epsilon}-\kappa \pi}{1-\kappa}} \\ &=\inf \sigma\left( \mathbb{Q}^0_{t,\epsilon}\right)+\frac{1}{2\alpha^2}\mathrm{Tr}_{\Pi L^2(\mathbb{R}^3)}\! \! \left[\sqrt{\frac{J_{t,\epsilon}-\kappa \pi}{1-\kappa}}-\sqrt{J_{t,\epsilon}}\right] \,.
\end{align*}
Using $J_{t,\epsilon}\pi=J_{t,\epsilon}$, and therefore $[J_{t,\epsilon},\pi]=0$, and again the fact that $J_{t,\epsilon}\geq c\, \pi$ for $\epsilon$ small enough and $\alpha$ large enough, as well as $|\sqrt{x}-\sqrt{y}|\leq \frac{1}{\sqrt{c}}|x-y|$ for $x\geq 0$ and $y\geq c$, we obtain for such $\epsilon,\alpha$, and $\kappa\leq c$
\begin{align*}
\pm &\mathrm{Tr}_{\Pi L^2(\mathbb{R}^3)}\! \! \left[\sqrt{\frac{J_{t,\epsilon}-\kappa \pi}{1-\kappa}}-\sqrt{J_{t,\epsilon}}\right]\leq \frac{1}{\sqrt{c}}\mathrm{Tr}_{\Pi L^2(\mathbb{R}^3)} \!\left|\frac{J_{t,\epsilon}-\kappa \pi}{1-\kappa}-J_{t,\epsilon}\right|\\
&= \frac{\kappa}{\sqrt{c}(1-\kappa)}\mathrm{Tr} \left|J_{t,\epsilon}-\pi\right|=\frac{\kappa(1+\epsilon)}{\sqrt{c}(1-\kappa)}\mathrm{Tr}\left[K^\mathrm{Pek}+\epsilon L^\mathrm{Pek}\right]\lesssim \frac{\kappa}{1-\kappa},
\end{align*}
where we have used that $K^\mathrm{Pek}$ and $L^\mathrm{Pek}$ defined in Definition \ref{Definition: Quadratic Approximation} are trace-class. Combining what we have so far  with the bound 
$$
\inf \sigma\left(\mathbb{Q}^0_{t,\epsilon}\right)\geq \frac{N}{2\alpha^2}-\frac{1}{2\alpha^2}\mathrm{Tr}\left[1-\sqrt{H^\mathrm{Pek}}\, \right]-D\left(\alpha^{-2}\epsilon+\alpha^{-\left(2+\frac{1}{5}\right)}\right)
$$ 
for small $\epsilon$, $|t|<\epsilon$ and large $\alpha$, and a suitable $D>0$, see \cite[Lemma B.5]{BS1}, yields
\begin{align*}
\inf_{|t|<\alpha^{-s}}\! \! \inf \sigma\left(\mathbb{Q}^{\alpha^{-w'}}_{t,\alpha^{-s}}\right)-\frac{N}{2\alpha^2}+\frac{1}{2\alpha^2}\mathrm{Tr}\left[1\! -\! \sqrt{H^\mathrm{Pek}}\, \right]\! \gtrsim -\! \left(\alpha^{-(2+s)}+\alpha^{-\left(2+\frac{1}{5}\right)}+\alpha^{-(2+w')}\right).
\end{align*}
In combination with Eq.~(\ref{Equation-MinimalBogoliubovEnergy}) we therefore obtain for a suitable $w>0$
\begin{align}
\nonumber
E_{\alpha,\Lambda}(\alpha^2p)  &\geq   e^\mathrm{Pek}_p-\frac{1}{2\alpha^2}\mathrm{Tr}\left[1-\sqrt{H^\mathrm{Pek}}\, \right]+O_{\alpha\rightarrow \infty}\Big(\alpha^{-(2+w)}\Big),
\end{align}
which concludes the proof of Eq.~(\ref{Equation-Main with cut-off}).

\appendix

\section{Auxiliary Results}
\label{Appedinx: Auxiliary Results}
\begin{lem}
\label{Lemma-IMSforGradient}
Let $g(k):=\! \chi^1\! \left(K^{-1}|k|\leq 2\right)\!k$ for $k\in \mathbb{R}$. Then there exists a constant $C>0$ such that for any bounded function $f:\mathbb{R}\rightarrow \mathbb{R}$ with $f'\in L^2\!\left(\mathbb{R}\right)$ and $K>0$, the double commutator is bounded by
\begin{align*}
\left\| \left[\left[g\left(\frac{1}{i}\frac{\mathrm{d}}{\mathrm{d}t}\right),f(t)\right],f(t)\right]\right\|_\mathrm{op} \leq C\| f'\|^2,
\end{align*}
where we write $f(t)$ for the multiplication operator with respect to the function $t\mapsto f(t)$. Furthermore we can choose the constant $C>0$ such that $\left\|[g(\frac{1}{i}\frac{\mathrm{d}}{\mathrm{d}t}),f(t)]\right\|_\mathrm{op}\leq C\sqrt{K}\|f'\|$.
\end{lem}
\begin{proof}
Let us start by defining the sequence $f_n(t):=\chi^1\left(\frac{|t|}{n}\leq 2 \right)f(t)$, which is compactly supported and therefore $f_n\in H^1\! \left(B_R(0)\right)$ by our assumptions. Hence there exist smooth and compactly supported $\widetilde{f}_n$ such that $\|f_n-\widetilde{f}_n\|_\infty+\|(f_n)'-(\widetilde{f}_n)'\|\underset{n\rightarrow \infty}{\longrightarrow}0$. Clearly the sequence $\widetilde{f}_n$ is uniformly bounded and approximates $f(t)$ in the strong operator topology, and consequently $ \left[\left[g\left(\frac{1}{i}\frac{\mathrm{d}}{\mathrm{d}t}\right),\widetilde{f}_n(t)\right],\widetilde{f}_n(t)\right]$ approximates $ \left[\left[g\left(\frac{1}{i}\frac{\mathrm{d}}{\mathrm{d}t}\right),f(t)\right],f(t)\right]$ in the strong operator topology as well. Hence $\left\| \left[\left[g\left(\frac{1}{i}\frac{\mathrm{d}}{\mathrm{d}t}\right),f(t)\right],f(t)\right]\right\|_\mathrm{op}$ is bounded from above by $\limsup_{n\rightarrow \infty}\left\| \left[\left[g\left(\frac{1}{i}\frac{\mathrm{d}}{\mathrm{d}t}\right),\widetilde{f}_n(t)\right],\widetilde{f}_n(t)\right]\right\|_\mathrm{op}$. Together with the observation $\| f'-(\widetilde{f}_n)'\|\underset{n\rightarrow \infty}{\longrightarrow}0$, we can therefore assume w.l.o.g. that $f$ is smooth and compactly supported. 

Going to Fourier space and defining $M(k,k'):=\sup_p \big|g\left(p+k+k'\right)-g\left(p+k\right)-g\left(p+k'\right)+g\left(p\right)\big|$, we can write 
\begin{align*}
& 2\pi\left\|\left[\left[g\left(\frac{1}{i}\frac{\mathrm{d}}{\mathrm{d}t}\right),f(t)\right],f(t)\right]\right\|_\mathrm{op}= \bigg\|\int \int  \widehat{f}(k)\widehat{f}(k')e^{it(k+k')}\bigg(g\left(\frac{1}{i}\frac{\mathrm{d}}{\mathrm{d}t}+k+k'\right)\\
&\ \ -g\left(\frac{1}{i}\frac{\mathrm{d}}{\mathrm{d}t}+k\right)-g\left(\frac{1}{i}\frac{\mathrm{d}}{\mathrm{d}t}+k'\right)+g\left(\frac{1}{i}\frac{\mathrm{d}}{\mathrm{d}t}\right)\bigg)\mathrm{d}k\mathrm{d}k'\bigg\|_\mathrm{op} \! \! \! \! \!  \leq \! \!\int\! \! \int  \big|\widehat{f}(k)\widehat{f}(k')\big|M(k,k')\mathrm{d}k\mathrm{d}k'\\
&=\! \! \int_{|k'|\leq K}\!  \int_{|k|\leq K}\!   \big|k\widehat{f}(k)k'\widehat{f}(k')\big|\frac{M(k,k')}{|k k'|}\mathrm{d}k\mathrm{d}k'\! +\! 2\! \int_{|k'|\leq K}\!  \int_{|k|> K}\!   \big|\widehat{f}(k)k'\widehat{f}(k')\big|\frac{M(k,k')}{|k'|}\mathrm{d}k\mathrm{d}k'\\
&\ \ \ \ \ \ +\int_{|k'|> K} \int_{|k|> K}  \big|\widehat{f}(k)\widehat{f}(k')\big|M(k,k')\mathrm{d}k\mathrm{d}k'.
\end{align*}
Making use of the fact that $\left|\frac{M(k,k')}{k k'}\right|\leq \|g''\|_\infty\lesssim \frac{1}{K}$, $\left|\frac{M(k,k')}{k'}\right|\leq 2\|g'\|_\infty\lesssim 1$ and $\left|M(k,k')\right|\leq 4\|g\|_\infty\lesssim K$, we obtain 
\begin{align*}
&2\pi\|[[g(\frac{1}{i}\frac{\mathrm{d}}{\mathrm{d}t}),f(t)],f(t)]\|_\mathrm{op}\lesssim \frac{1}{K}\left(\int_{|k|\leq K}\big|k\widehat{f}(k)\big|\mathrm{d}k\right)^2\\
&\ \ \ \ +2\int_{|k'|\leq K}\big|k'\widehat{f}(k')\big|\mathrm{d}k'\int_{|k|> K}\frac{1}{|k|}\big|k\widehat{f}(k)\big|\mathrm{d}k+K\left(\int_{|k|> K}\frac{1}{|k|}\big|k\widehat{f}(k)\big|\mathrm{d}k\right)^2\\
&\leq \frac{2}{K}\left(\int_{|k|\leq K}\big|k\widehat{f}(k)\big|\mathrm{d}k\right)^2+2K\left(\int_{|k|> K}\frac{1}{|k|}\big|k\widehat{f}(k)\big|\mathrm{d}k\right)^2\\
&\leq \|f'\|^2\left(\frac{2}{K}\int_{|k|\leq K} \mathrm{d}k+2K\int_{|k|>K}\frac{1}{|k|^2}\mathrm{d}k\right)\leq 8\|f'\|^2.
\end{align*}

In order to estimate the operator norm of $[g(\frac{1}{i}\frac{\mathrm{d}}{\mathrm{d}t}),f(t)]$, we can assume as above that $f$ is smooth and compactly supported. We compute
\begin{align*}
\sqrt{2\pi}&\left\|[g\left(\frac{1}{i}\frac{\mathrm{d}}{\mathrm{d}t}\right),f(t)]\right\|_\mathrm{op}\leq \int \big|\widehat{f}(k)\big| \left\|g\left(\frac{1}{i}\frac{\mathrm{d}}{\mathrm{d}t}+k\right)-g\left(\frac{1}{i}\frac{\mathrm{d}}{\mathrm{d}t}\right)\right\|_\mathrm{op} \mathrm{d}k\\
&\leq \left\|g'\right\|_\infty\int_{|k|\leq K} \big|k\widehat{f}(k)\big|  \mathrm{d}k+2\left\|g\right\|_\infty \int_{|k|> K} \frac{1}{|k|}\big|k\widehat{f}(k)\big|  \mathrm{d}k\\
&\leq \sqrt{2K}\left\|g'\right\|_\infty \|f'\|+\sqrt{\frac{8}{K}}\left\|g\right\|_\infty \|f'\|.
\end{align*}
Using $\left\|g'\right\|_\infty\lesssim 1$ and $\left\|g\right\|_\infty\lesssim K$ concludes the proof.
\end{proof}

\begin{lem}
\label{Lemma: Momentum Cut-off}
For $K>0$ we have the estimate $\|\chi\left(|\nabla|> K\right)\nabla \varphi^\mathrm{Pek}\|\lesssim \frac{1}{\sqrt{K}}$.
\end{lem}
\begin{proof}
We can write $\varphi^\mathrm{Pek}=4\sqrt{\pi}\left(-\Delta\right)^{-\frac{1}{2}}\left|\psi^\mathrm{Pek}\right|^2$, where $\psi^\mathrm{Pek}$ is as in Definition \ref{Definition: Quadratic Approximation}. Hence the Fourier transform of $\nabla \varphi^\mathrm{Pek}$ reads $\widehat{\nabla \varphi^\mathrm{Pek}}(k)=\frac{ik}{|k|}\widehat{\left|\psi^\mathrm{Pek}\right|^2}(k)$, and therefore
\begin{align*}
\|\chi\left(|\nabla|> K\right)\nabla \varphi^\mathrm{Pek}\|^2=\int_{|k|>K}\left|\widehat{\left|\psi^\mathrm{Pek}\right|^2}(k)\right|^2\mathrm{d}k\leq \left\| |k|^2 \widehat{\left|\psi^\mathrm{Pek}\right|^2}(k)\right\|_\infty^2 \int_{|k|>K}\frac{1}{|k|^4}\mathrm{d}k\lesssim \frac{1}{K},
\end{align*}
where we used $\psi^\mathrm{Pek}\in H^2\! \left(\mathbb{R}^3\right)$ and consequently $\left\| |k|^2 \widehat{\left|\psi^\mathrm{Pek}\right|^2}(k)\right\|_\infty<\infty$. 
\end{proof}

\begin{lem}
\label{Lemma: Compatibility}
With $\Pi$  the projection defined in Definition \ref{Definition: Finite dimensional Projection}, we have
\begin{align*}
\left\|[|\nabla|,\Pi]\right\|_{\mathrm{op}}\lesssim \alpha^{-4(1+\sigma)}.
\end{align*}
\end{lem}
\begin{proof}
Using the Fourier transformation, we can write $\widehat{\Pi \varphi}(k)=\sum_{n=1}^N \braket{f_n|\widehat{\varphi}} f_n(k)$, with the help of non-negative functions $f_n$ having pairwise disjoint support, which additionally satisfy $\|f_n\|=1$ and $\mathrm{supp}\left(f_n\right)\subset B_{\sqrt{3}\alpha^{-4(1+\sigma)}}\left(z^n\right)$ for some $z^n \in \mathbb{R}^3$. Therefore 
\begin{align*}
\widehat{[|\nabla|,\Pi]\varphi}(k)\! =\! \sum_{n=1}^N\!  \left(\braket{f_n|\widehat{\varphi}}|k|\! -\! \Big\langle f_n\Big |\widehat{|\nabla|\varphi}\Big\rangle \right)\! f_n(k)\! =\! \sum_{n=1}^N \int\!  f_n(k')\widehat{\varphi}(k')\! \left(|k|\! -\! |k'|\right) \mathrm{d}k' f_n(k).
\end{align*}
Using that the functions $f_n$ have disjoint support, as well as the fact that $\left||k|-|k'|\right|\leq 2\sqrt{3}\alpha^{-4(1+\sigma)}$ for $k,k'\in \mathrm{supp}\left(f_n\right)$, we obtain furthermore
\begin{align*}
&\left\|[|\nabla|,\Pi]\varphi\right\|^2=\sum_{n=1}^N \int \left|\int\!  f_n(k')\widehat{\varphi}(k')\! \left(|k|\! -\! |k'|\right) \mathrm{d}k'\right|^2 |f_n(k)|^2\mathrm{d}k\\
&\ \ \ \ \leq 12\alpha^{-8(1+\sigma)}\sum_{n=1}^N \left|\int\!  f_n(k')\left|\widehat{\varphi}(k')\right| \mathrm{d}k'\right|^2\leq 12\alpha^{-8(1+\sigma)}\left\| |\widehat{\varphi}|\right\|^2=12\alpha^{-8(1+\sigma)}\|\varphi\|^2,
\end{align*}
where we have used that $f_n$ is an orthonormal system.
\end{proof}

\bibliographystyle{plain}

\end{document}